\documentclass[a4paper,12pt]{amsart} 
\usepackage[lmargin=3.7cm,rmargin=2.7cm,bmargin=3.5cm,tmargin=3cm]{geometry}

\usepackage{color}
\usepackage{booktabs}
\usepackage{mathrsfs}
\usepackage{amsthm}
\usepackage{amsmath}
\usepackage{amssymb}
\usepackage{esint}
\usepackage[numbers]{natbib}
\usepackage{verbatim}
\usepackage{bbm}

\makeatletter

\usepackage{enumitem}		
  \theoremstyle{plain}
  \newtheorem{thm}{\protect\theoremname}
  \theoremstyle{plain}
  \newtheorem{lyxalgorithm}[thm]{\protect\algorithmname}
  \theoremstyle{definition}
  \newtheorem{defn}[thm]{\protect\definitionname}
  \theoremstyle{remark}
  \newtheorem{rem}[thm]{\protect\remarkname}
  \theoremstyle{plain}
  \newtheorem{lem}[thm]{\protect\lemmaname}
 \theoremstyle{definition}
  \newtheorem{example}[thm]{\protect\examplename}
  \theoremstyle{plain}
  \newtheorem{cor}[thm]{\protect\corollaryname}
  \theoremstyle{plain}
  \newtheorem{prop}[thm]{\protect\propositionname}

  \providecommand{\algorithmname}{Algorithm}
  \providecommand{\definitionname}{Definition}
  \providecommand{\examplename}{Example}
  \providecommand{\lemmaname}{Lemma}
  \providecommand{\propositionname}{Proposition}
  \providecommand{\remarkname}{Remark}
\providecommand{\corollaryname}{Corollary}
\providecommand{\theoremname}{Theorem}

\newcommand{\defeq}{\mathrel{\mathop:}=}

\newcommand{\uarg}{\,\cdot\,}
\newcommand{\ud}{\mathrm{d}}
\newcommand{\R}{\mathbb{R}}

\newcommand{\E}{\mathbb{E}}
\newcommand{\B}{\mathcal{B}}

\newcommand{\var}{\mathrm{var}}

\newcommand{\lecx}{\le_{\mathrm{cx}}}

\makeatother

\title{Establishing some order amongst exact approximations of MCMCs }
\author{Christophe Andrieu}
\address{Christophe Andrieu, School of Mathematics, University of Bristol, BS8 1TW, United
Kingdom}
\email{C.Andrieu@bristol.ac.uk}

\author{Matti Vihola}
\address{Matti Vihola, Department of Mathematics and Statistics, University of
Jyväskylä P.O.Box 35, FI-40014 Univ.~of Jyväskylä, Finland}
\email{matti.vihola@iki.fi}



\begin{document}

\begin{abstract} 
Exact approximations of Markov chain Monte Carlo (MCMC) algorithms
are a general emerging class of sampling algorithms. One of the main
ideas behind exact approximations consists of replacing intractable
quantities required to run standard MCMC algorithms, such as the target
probability density in a Metropolis-Hastings algorithm, with estimators.
Perhaps surprisingly, such approximations lead to powerful algorithms
which are exact in the sense that they are guaranteed to have correct
limiting distributions. In this paper we discover a general framework
which allows one to compare, or order, performance measures of two
implementations of such algorithms. In particular, we establish an
order with respect to the mean acceptance probability, the first autocorrelation
coefficient, the asymptotic variance and the right spectral
gap. The key notion to guarantee the ordering is that of the convex
order between estimators used to implement the algorithms. We believe
that our convex order condition is close to optimal, and this is supported
by a counter-example which shows that a weaker variance
order is not sufficient. The
convex order plays a central role by allowing us to construct a martingale
coupling which enables the comparison of performance measures of Markov
chain with differing invariant distributions, contrary to existing results. 
We detail applications of our result
by identifying extremal distributions within given classes of approximations,
by showing that averaging replicas improves performance in a monotonic
fashion and that stratification is guaranteed to improve performance
for the standard implementation of the Approximate Bayesian Computation
(ABC) MCMC method. 
\end{abstract}

\sloppy
\maketitle

\section{Introduction}
\label{sec:Introduction} 

Consider a probability distribution $\pi$ defined on some measurable
space $\bigl(\mathsf{X},\mathcal{X}\bigr)$ and assume that sampling
realisations from this probability distribution is of interest. A
generic and popular way of achieving this consists of using Markov
chain Monte Carlo algorithms (MCMC), of which the Metropolis-Hastings
(MH) algorithm is the main workhorse. 

Let $\{q(x,\cdot)\}{}_{x\in\mathsf{X}}$
be a family of probability distributions on $\bigl(\mathsf{X},\mathcal{X}\bigr)$
such that $x\mapsto q(x,A)$ is measurable for all $A\in \mathcal{X}$;
in other words $(x,A)\mapsto q(x,A)$ defines a kernel.
(We assume hereafter implicitly that such a measurability condition is
satisfied by given families of probability measures, and 
that given functions are measurable.)
For any $x,y\in\mathsf{X}$
we define the ``acceptance ratio'' $r\bigl(x,y\bigr)$ as follows:
for $(x,y)\in\mathsf{R}\subset\mathsf{X}^{2}$ (with $\mathsf{R}$
the symmetric set as defined in \citep[Proposition 1]{tierney-note})
we let $r\bigl(x,y\bigr)$ be the Radon-Nikodym derivative 
\[
0<r(x,y):=\frac{\pi(\mathrm{d}y)q(y,\mathrm{d}x)}{\pi(\mathrm{d}x)q(x,\mathrm{d}y)}<\infty,
\]
and $r(x,y)=0$ otherwise. The MH algorithm defines a Markov chain
with the following transition kernel
\begin{equation}
P(x,{\rm d}y)=q(x,{\rm d}y)\min\left\{ 1,r(x,y)\right\}
+\delta_{x}\left({\rm d}y\right)\rho(x),
\label{eq:marginalMH}
\end{equation}
where
\[
\rho(x):=1-\int\min\left\{ 1,r(x,y)\right\} q(x,{\rm d}y),
\]
and $\delta_{x}$ is the Dirac measure at $x\in\mathsf{X}$. It is
standard to show that $P$ is reversible with respect to $\pi$ and
hence leaves $\pi$ invariant.

In some situations, evaluation of the ratio $r\bigl(x,y\bigr)$ is
either impossible or overly expensive, therefore rendering the algorithm
non-viable. A canonical example is when $\pi$ is the marginal of
a probability density, say $\pi(x)=\int\bar{\pi}(x,z)\mathrm{d}z$
for some latent variables $z$, where the integral is intractable
and the ``marginal MCMC'' targeting $\pi$ is therefore not implementable.
A classical way of addressing this problem consists of running an
MCMC targeting the joint distribution $\bar{\pi}$, which may however
become prohibitively inefficient in situations where the size of the
latent variable is high--this is for example the case in general state-space
models \citep{andrieu-doucet-holenstein}. A powerful alternative
which has recently attracted substantial interest consists, in simple
terms, of replacing the value of $\pi(x)$ with a noisy, but computationally
cheap, measurement whenever it is required in the implementation of
the MH algorithm above. Although this idea may at first appear naive,
it turns out to lead to exact algorithms, that is sampling from $\pi$
is guaranteed at equilibrium under mild assumptions on the nature
of the noise involved.

We now describe in more detail a modification of the ``exact'',
or marginal, MH algorithm above, where one replaces the density values
$\pi(x)$ with estimators $\hat{\pi}\bigl(x\bigr)$ for all $x\in\mathsf{X}$.
The estimators $\hat{\pi}\bigl(x\bigr)$ are assumed to be non-negative
and unbiased (up to a constant), in a sense that there exists $C>0$
such that for all $x\in\mathsf{X}$, $\mathbb{E}[\hat{\pi}\bigl(x\bigr)]=C\times\pi(x)$.
To fix ideas we first present the method in an algorithmic form, and
later define precisely the related mathematical construction. We call
this method a pseudo-marginal approximation of $P$ or simply the
pseudo-marginal algorithm, following \citep{andrieu-roberts,andrieu-vihola-2012}.
\begin{lyxalgorithm}[Pseudo-marginal algorithm] 
\label{alg:pseudo-marginal}Assume $(X_{0},\hat{\pi}\bigl(X_{0}\bigr))\in\mathsf{X}\times(0,\infty)$
are some initial values. Then, for $n=1,2,\ldots$ 
\begin{enumerate}[label=(\alph*)]
\item Propose a transition $Y_{n}\sim q(X_{n},\uarg).$
\item Given $Y_{n},$ generate a random variable $\hat{\pi}(Y_{n})\ge0$
(such that $\mathbb{E}[\hat{\pi}(Y_{n})]=C\pi(Y_{n})$).
\item With probability
  \[
      \min\bigg\{1,\frac{\hat{\pi}(Y_{n})}{\hat{\pi}(X_{n-1})}\frac{q(Y_{n,}X_{n-1})}{q(X_{n-1},Y_{n})}\bigg\},
  \]
accept and set
$\big(X_{n},\hat{\pi}(X_{n})\big)=\big(Y_{n,}\hat{\pi}(Y_{n})\big)$;
otherwise reject and set $\big(X_{n},\hat{\pi}(X_{n})\big)
=\big(X_{n-1},\hat{\pi}(X_{n-1})\big)$.
\end{enumerate}
\end{lyxalgorithm} 
For example, in the situation where $\pi(x)=\int\bar{\pi}(x,z)\mathrm{d}z$
one could use an importance sampling estimator of the integral and
define
\[
\hat{\pi}\bigl(x\bigr)=\frac{1}{N}\sum_{i=1}^{N}\frac{\bar{\pi}\bigl(x,Z_{i}\bigr)}{g_{x}(Z_{i})},
\]
where $Z_{1},\ldots,Z_{N}$ are samples from some instrumental probability
density $g_{x}$. The algorithm resulting from this choice is known
as the grouped independence Metropolis-Hastings \citep{beaumont}
and has proved useful in the context of inference in phylogenetic
trees. Other important instances of pseudo-marginal algorithms include
approximate Bayesian computation (ABC) MCMC \citep{marjoram-molitor-plagnol-tavare}
(see also Section \ref{sub:Stratification}), the particle marginal
MH \citep{andrieu-doucet-holenstein} and algorithms for inference
in diffusions \citep{beskos-papaspiliopoulos-roberts-fearnhead}.
Note that Algorithm \ref{alg:pseudo-marginal} is intact if we multiplied
all our estimators by a constant $C'>0$. Therefore, without loss
of generality, we will set $C=1$ in the remainder of the paper for
notational simplicity.

We now turn to an abstract representation of the pseudo-marginal algorithm,
used in the remainder of this paper, which has the advantage that
it simplifies notation and highlights the key structure and quantities
underpinning their behaviour \citep{andrieu-roberts,andrieu-vihola-2012,sherlock-thiery-roberts-rosenthal}.
The main idea is to introduce the normalised estimator $\varpi(x):=\hat{\pi}(x)/\pi(x)$
and to view this as a multiplicative perturbation, or noise, of the
true density $\pi(x)$, since then $\hat{\pi}(x)=\pi(x)\varpi(x)$.
An immediate consequence is that with $W_{n-1}:=\varpi(X_{n-1})$
and $U_{n}:=\varpi(Y_{n})$ the acceptance ratio in Algorithm \ref{alg:pseudo-marginal}
takes the form $r(X_{n-1},Y_{n})U_{n}/W_{n-1}$ which highlights the
structure of this approximation, a multiplicative perturbation of
the acceptance ratio of the exact algorithm, and the fact that $\bigl\{ X_{n},W_{n}\bigr\}$
is a Markov chain. This leads us to the following abstract model.
Let $\{Q_{x}\}_{x\in\mathsf{X}}$ be a family of probability measures
on the non-negative reals $\bigl(\R_{+},\mathcal{B}(\R_{+})\bigr)$
indexed by $x\in\mathsf{X}$ and such that $\int wQ_{x}(\mathrm{d}w)=1$
for any $x\in\mathsf{X}$. Consider the following probability distribution
on $\bigl(\mathsf{X}\times\R_{+},\mathcal{X}\times\mathcal{B}(\R_{+})\bigr)$
defined through

\[
\tilde{\pi}(\mathrm{d}x\times\mathrm{d}w):=\pi(\mathrm{d}x)Q_{x}(\mathrm{d}w)w.
\]
This formalises the multiplicative perturbation idea and includes
the distribution of the estimators in the target distribution. Now
one can check that the Markov transition probability $\tilde{P}$
of the pseudo-marginal approximation of the marginal kernel $P$ is
a MH algorithm targeting $\tilde{\pi}$ and can be written as follows
for any $(x,w)\in\mathsf{X}\times(0,\infty)$ \citep{andrieu-roberts,andrieu-vihola-2012}
\begin{align}
\label{eq:kernelpseudomarginal}
\tilde{P}(x,w;\mathrm{d}y\times\mathrm{d}u)&:=q(x,\mathrm{d}y)Q_{y}(\mathrm{d}u)\min\left\{
1,r(x,y)\frac{u}{w}\right\} \\
&\phantom{:=}+\delta_{x,w}(\mathrm{d}y\times\mathrm{d}u)\tilde{\rho}(x,w),\nonumber
\end{align}
with rejection probability

\[
\tilde{\rho}(x,w):=1-\int\min\left\{ 1,r(x,y)\frac{u}{w}\right\} q(x,\mathrm{d}y)Q_{y}(\mathrm{d}u).
\]
Note that while the distribution of the normalised estimators $Q_{y}(\mathrm{d}u)$
appears in the kernel and the rejection probability, it does not appear
explicitly in the algorithm. The fundamental marginal identity $\tilde{\pi}(\mathrm{d}x\times\mathbb{R}_{+})=\pi(\mathrm{d}x)$
is important since it tells us that despite being approximations of
$P$, pseudo-marginal algorithms are exact in the sense that they
sample marginally from the desired distribution $\pi$ at equilibrium.

Often in practice there are various possible choices for the estimators
in Algorithm \ref{alg:pseudo-marginal}, or equivalently the family
$\{Q_{x}\}_{x\in\mathsf{X}}$, depending on the specific application.
For example, in case of the simple importance sampling estimator presented
above one may choose between various families of importance sampling
distributions $\{g_{x}\}_{x\in\mathsf{X}}$, or vary the number of
samples $N$. The ultimate aim of the present work is to develop simple
and general tools for the characterisation of the performance of pseudo-marginal
algorithms as a function of the family $\{Q_{x}\}_{x\in\mathsf{X}}$,
therefore allowing comparisons. We refer the reader to \citep{andrieu-vihola-2012}
for other quantitative and qualitative properties of such algorithms.

We now recall standard performance measures relevant to the MCMC
context. For a generic Markov transition kernel $\Pi$ with invariant
distribution $\mu$, both defined on some measurable space
$\left(\mathsf{E},\mathcal{E}\right)$, and any function
$f:\mathsf{E}\rightarrow\mathbb{R}$, we define the asymptotic variance
of $f$ for $\Pi$ as follows. Denote
$\mathbb{E}_{\mu}\bigl(\cdot\bigr)$ and ${\rm var_{\mu}(\cdot)}$ the
expectation and variance operators corresponding to the Markov chain
$\{\Phi_{i}\}_{i\ge0}$ with transition kernel $\Pi$ and such that
$\Phi_{0}\sim\mu$, and denote ${\rm var}_{\mu}(f):={\rm
var}_{\mu}\bigl(f(\Phi_{0})\bigr)$). Then the asymptotic variance is
defined as 
\begin{equation*}
\mathrm{var}(f,\Pi)\defeq 
\lim_{M\rightarrow\infty}\mathrm{var}_{\mu}\left(
  M^{-1/2}{\textstyle \sum}_{i=1}^{M}f(\Phi_{i})\right),
\end{equation*}
where the limit is guaranteed to exist for reversible Markov chains,
but may be infinite \citep[cf.~][]{tierney-note}. When finite, 
the asymptotic variance naturally appears in the central limit theorem
\cite[e.g.][]{jones}, but it also
characterises the finite sample efficiency of MCMC
algorithms; see for instance the finite sample bounds given in
\cite{latuszynski-miasojedow-niemiro,rudolf}. When $\Pi$ is
reversible with respect to $\mu$, the so-called right spectral gap
${\rm Gap}_{R}\bigl(\Pi\bigr)$, defined precisely in Section \ref{sec:pseudo-marginals-results},
turns out in some scenarios to be an indicative criterion of performance.
For instance, when $\Pi$ is
a positive operator, the right spectral gap is also the absolute spectral
gap, which characterises the geometric convergence to
equilibrium \cite{roberts-tweedie}.  

Before focusing on the comparison of pseudo-marginal algorithms we
briefly review here what is known about the standard MH algorithm,
since the result will be referred to in several places in this paper
and helps to motivate our work. The result is in its first form due
to Peskun \citep{peskun} and was later extended to more general set-ups
in \citep{caracciolo-pelissetto-sokal,tierney-note,leisen-mira}.
\begin{thm}[Peskun] 
\label{th:peskun}
Let $\mu$ be a probability distribution on some
measurable space $\left(\mathsf{E},\mathcal{E}\right)$ and let $\Pi_{1},\Pi_{2}$
be two $\mu-$reversible Markov transition kernels. Assume that for
any $x\in\mathsf{E}$ and any $A\in\mathcal{E}$ such that $x\notin A$,
the transitions satisfy $\Pi_{1}(x,A)\geq\Pi_{2}(x,A)$. 
\begin{enumerate}[label=(\alph*)]
\item If $f:\mathsf{X}\rightarrow\mathbb{R}$ satisfies ${\rm var}_{\mu}(f)<\infty$,
then $\mathrm{var}(f,\Pi_{1})\leq\mathrm{var}(f,\Pi_{2})$.
\item The right spectral gaps satisfy ${\rm Gap}_{R}\bigl(\Pi_{1}\bigr)\geq{\rm {\rm Gap}}_{R}\bigl(\Pi_{2}\bigr)$. 
\end{enumerate}
\end{thm} 
In fact, as pointed out by \citep{caracciolo-pelissetto-sokal}, the
off-diagonal order $\Pi_{1}(x,A)\geq\Pi_{2}(x,A)$ is stronger than
needed for Theorem \ref{th:peskun} to hold, and a weaker condition
is that 
\[
\int f\bigl(x\bigr)f\bigl(y\bigr)\mu\bigl({\rm d}x\bigr)\Pi_{1}\bigl(x,{\rm d}y\bigr)\leq\int f\bigl(x\bigr)f\bigl(y\bigr)\mu\bigl({\rm d}x\bigr)\Pi_{2}\bigl(x,{\rm d}y\bigr),
\]
for any $f$ such that ${\rm var}_{\mu}(f)<\infty$. This condition
turns out to be a necessary and sufficient condition. However the
popularity of Theorem \ref{th:peskun} stems from the simplicity of
its statement providing a simple and intuitive criterion for the comparison
of performance of algorithms, which can be checked in practice. As
we shall see, Peskun's result is however not directly relevant to
pseudo-marginal algorithms when the aim is to compare different approximation
strategies. This stems from the fact that changing $\{Q_{x}\}_{x\in\mathsf{X}}$
changes the invariant distribution of the Markov transition kernel
involved.

More specifically, consider two families of distributions, $\bigl\{ Q_{x}^{(1)}\bigr\}_{x\in\mathsf{X}}$
and $\bigl\{ Q_{x}^{(2)}\bigr\}_{x\in\mathsf{X}}$ leading to two
competing pseudo-marginal approximations of $P$, say $\tilde{P}_{1}$
and $\tilde{P}_{2}$, with distinct invariant distributions $\tilde{\pi}_{1}$
and $\tilde{\pi}_{2}$ respectively. Note that both algorithms target
$\pi(\cdot)$ marginally and share the same family of proposal distributions
$\{q(x,\cdot)\}_{x\in\mathsf{X}}$. In the light of Peskun's result,
a natural question is to find useful conditions on the families $\bigl\{ Q_{x}^{(1)}\bigr\}_{x\in\mathsf{X}}$
and $\bigl\{ Q_{x}^{(2)}\bigr\}_{x\in\mathsf{X}}$ which would ensure
that $\mathrm{var}(f,\tilde{P}_{1})\leq\mathrm{var}(f,\tilde{P}_{2})$
for certain classes of functions $f:\mathsf{X}\times\R_{+}\rightarrow\mathbb{R}$,
or that ${\rm Gap}_{R}\bigl(\tilde{P}_{1}\bigr)\leq{\rm {\rm Gap}}_{R}\bigl(\tilde{P}_{2}\bigr)$. 

As we shall see, despite the difficulty pointed out earlier, it is
possible to prove such results, assuming a classical convex stochastic
order between the unit mean random variables $W_{x}^{(1)}\sim Q_{x}^{(1)}$
and $W_{x}^{(2)}\sim Q_{x}^{(2)}$ for $x\in\mathsf{X}$ (see Theorem
\ref{thm:orderingpseudomarginals}). Because the convex order is stronger
than the simpler variance order between $W_{x}^{(1)}$ and $W_{x}^{(2)}$,
one may wonder whether the variance order could be sufficient to imply
a similar result. We show, by means of a counterexample, that the
variance order is not sufficient to guarantee such results (see Example
\ref{ex:counterexamplevarianceorder}). 

Our framework benefits from a considerable body of work on the convex
order (see \citep{mullercomparison,shaked-shanthikumar} for recent
reviews) which allows one to import known results and establish ordering
results for pseudo-marginal algorithms with minimal effort. Some applications
are discussed in some detail in Sections \ref{sec:Averaging-in-the},
\ref{sub:Stratification} and \ref{sec:Extremal-properties}. We also
point to a very recent application of our result, behind an interesting
reasoning aimed at deriving a type of quantitative bounds for pseudo-marginal
algorithms in some contexts \citep[Theorem 3]{2014arXiv1404.6298B}.

We do not know of earlier works with similar Peskun-type orders for
Markov chains with different invariant distributions. The recent work
\citep{maire-douc-olsson} establishes Peskun-type orders for inhomogeneous
Markov chains, which is applicable to some algorithms which are discussed
in the present work as well. However, we emphasise that our results
are much more general when applied in the context of pseudo-marginal
algorithms.

The remainder of the paper is organised as follows. In section \ref{sec:Exact-approximations-ratios},
before embarking on the analysis of the pseudo-marginal algorithms
described above, we first focus on related and practically relevant
algorithms for which convex ordering is pertinent but the comparison
is mathematically much simpler. This allows a gentle and progressive
introduction of the material, and helps to explain also why pseudo-marginal
algorithms are much more difficult to analyse. In Section \ref{sec:pseudo-marginals-results},
we formulate our main findings whose proofs are postponed to Section
\ref{sec:Proofs}. We prepare the proofs by recording variational
bounds for the difference of asymptotic variances in Section \ref{sec:Variational-bounds-for}
which we could not find formulated elsewhere in the literature. In
Section \ref{sec:Proofs} we develop our main arguments, which rely
on an embedding of the probability distributions involved into a single
probability distribution: this is possible thanks to a martingale
coupling which is itself a by-product of the convex order. In Section
\ref{sec:Applications}, we illustrate the usefulness of the framework
in practice. We conclude by discussing other applications and extensions
of the present work in Section \ref{sec:Discussion-and-perspectives}.


\section{Convex order and a simple application}
\label{sec:Exact-approximations-ratios} 

In this section we briefly review well known equivalent characterisations
of the convex order which we will use in the remainder of the paper.
We also provide an example of algorithms, related to pseudo-marginal
algorithms but different in an essential manner, for which the notion
of convex order characterises performance without the need of more
sophisticated mathematical developments.

The convex order is a natural way of comparing the ``variability''
or ``dispersion'' of two random variables or distributions \citep{shaked-shanthikumar,mullercomparison}.
\begin{defn}
\label{def:cx} The random variables $W_{1}\sim F_{1}$ and $W_{2}\sim F_{2}$
are \emph{convex ordered}, denoted $W_{1}\lecx W_{2}$ or $F_{1}\lecx F_{2}$
hereafter, if for any convex function $\phi:\mathbb{R}\to\mathbb{R}$,
\[
\mathbb{E}[\phi(W_{1})]\le\mathbb{E}[\phi(W_{2})],
\]
 whenever the expectations are well-defined.\end{defn}
\begin{rem}
\label{rem:cx-vs-var}For integrable $W_{1}$ and $W_{2}$, the convex
order $W_{1}\lecx W_{2}$ clearly implies $\mathbb{E}[W_{1}]=\mathbb{E}[W_{2}]$
from the convexity of $x\mapsto x$ and $x\mapsto-x$, and if $W_{1}$
and $W_{2}$ are square integrable, then ${\rm var}(W_{1})\le{\rm var}(W_{2})$,
since $x\mapsto x^{2}$ is convex. The converse, however, does not
generally hold true. This turns out to be an important point when
discussing the characterisation of performance of pseudo-marginal
approximations.\end{rem}
\begin{lem}
\label{lem:cx-char} Suppose that $\mathbb{E}[W_{1}]=\mathbb{E}[W_{2}]\in\mathbb{R}$.
Then, $W_{1}\lecx W_{2}$ is equivalent to: 
\begin{enumerate}[label=(\alph*)]
\item \label{item:cx-plus} $\mathbb{E}[(W_{1}-t)_{+}]\le\mathbb{E}[(W_{2}-t)_{+}]$
for all $t\in\mathbb{R}$, where $(\,\cdot\,)_{+}:=\max\{\,\cdot\,,0\}$. 
\item \label{item:cx-min} $\mathbb{E}[\min\{W_{1},t\}]\ge\mathbb{E}[\min\{W_{2},t\}]$
for all $t\in\mathbb{R}$.
\end{enumerate}
\end{lem}
\begin{proof}
Condition \ref{item:cx-plus} is a characterisation of the \emph{increasing
convex order} (Definition \ref{def:cx} restricted to non-decreasing
convex $\phi$) \citep[Theorem 4.A.2]{shaked-shanthikumar}, and the
increasing convex order with identical expectations implies the convex
order \citep[Theorem 4.A.35]{shaked-shanthikumar}. Similarly, \ref{item:cx-min}
implies $\mathbb{E}[(W_{1}-t)_{-}]\ge\mathbb{E}[(W_{2}-t)_{-}]$ which
is equivalent to the \emph{increasing concave order} (Definition \ref{def:cx}
with non-decreasing concave functions $\phi$), implying the desired
convex order \citep[Theorem 4.A.35]{shaked-shanthikumar}. \end{proof}
\begin{rem}
If $-\infty\le\underline{w}\le\bar{w}\le\infty$ are constants such
that $W_{1},W_{2}\in[\underline{w},\bar{w}]$ almost surely, then
it is sufficient to consider the conditions in Lemma \ref{lem:cx-char}
\ref{item:cx-plus} or \ref{item:cx-min} with $t\in[\underline{w},\bar{w}]$
only.
\end{rem}
It should be clear by application of Jensen's inequality that if $W_{1}$
and $W_{2}$ are defined on the same probability space and $\mathbb{E}[W_{2}\mid W_{1}]=W_{1}$,
then $W_{1}\lecx W_{2}$. The following characterisation of the
convex order, often referred to as Strassen's theorem \citep[Theorem 8]{strassen},
establishes the converse, that is that the convex order implies the
existence of this type of martingale representation. This characterisation
turns out to be central to our analysis as it allows us to eventually
``embed'' $\tilde{\pi}_{1}$ and $\tilde{\pi}_{2}$ into a unique
probability distribution and open up the possibility to use Hilbert
space techniques on a common space.
\begin{thm}
\label{th:strassen}Suppose that $\mathbb{E}[W_{1}]$ and $\mathbb{E}[W_{2}]$
are well-defined. Then, $W_{1}\lecx W_{2}$ if and only if there
exists a probability space with random variables $\check{W}_{1}$
and $\check{W}_{2}$ coinciding with $W_{1}$ and $W_{2}$ in distribution
respectively, such that $(\check{W}_{1},\check{W}_{2})$ is a martingale
pair, that is, $\mathbb{E}[\check{W}_{2}\mid\check{W}_{1}]=\check{W}_{1}$
(a.s.). 
\end{thm}
Before turning to the study of pseudo-marginal algorithms, we consider
first another class of related algorithms, which are much simpler
to analyse, yet practically relevant \citep{ceperley1999penalty,nicholls2012coupled,karagiannis-andrieu-2013}.
These algorithms are closely related to pseudo-marginal algorithms
in that they use noisy measurements of the acceptance ratio $r(x,y)$
of $P$, but the way the approximation is obtained differs in a fundamental
way. For pseudo-marginal algorithms the approximation of $r(x,y)$
stems at each iteration from a previously sampled, and ``recycled'',
approximation of $\pi$ at $x\in\mathsf{X}$ and a fresh approximation
of $\pi$ at $y\in\mathsf{X}$ (see Algorithm \ref{alg:pseudo-marginal}).
In contrast for the algorithms we are now concerned with $r(x,y)$
is approximated afresh whenever it is needed. More precisely, with
$(\mathsf{X},\mathcal{X})$ and $\pi$ as earlier and using a formalism
similar to that used for pseudo-marginal algorithms, we consider the
following Markov transition probability on $(\mathsf{X},\mathcal{X})$,
\[
\mathring{P}(x;\mathrm{d}y)=q(x,\mathrm{d}y)\int\min\left\{ 1,r(x,y)\varpi\right\} Q_{xy}(\mathrm{d}\varpi)+\delta_{x}(\mathrm{d}y)\mathring{\rho}(x),
\]
where $\{Q_{xy}\}{}_{(x,y)\in\mathsf{X}^{2}}$ is a family of probability
measures on positive reals, $r(x,y)$ is as in Section \ref{sec:Introduction}
and $\mathring{\rho}(x)\in[0,1]$ ensures that this is a Markov transition
probability. We stress again on the fact that while $r(x,y)$ is intractable,
the product $r(x,y)\varpi$ is assumed to be computable. It can be
shown that the condition 
\[
\int Q_{xy}(\mathrm{d}\varpi)\varpi\mathbb{I}\left\{ \varpi\in A\right\} 
=\int Q_{yx}(\mathrm{d}\varpi)\mathbb{I}\left\{ \frac{1}{\varpi}\in A\right\}
\]
for all $x,y\in\mathsf{X}$ and all
$A\in\B(\R_+)$
ensures that $\mathring{P}$ is reversible with respect to $\pi;$
Lemma \ref{lem:circlePreversible}
in Appendix \ref{sec:perturbed-mh} details 
a slightly more general statement.
\begin{example} 
\label{ex:toyP_ring} 
Suppose $a_{xy}=a_{yx}>0$ for all $x,y\in\mathsf{X}$,
then the distributions
\[
Q_{xy}(\ud \varpi)\defeq 
\frac{\delta_{a_{xy}}\left({\rm d}\varpi\right)+a_{xy}\delta_{a_{xy}^{-1}}\left({\rm
d}\varpi\right)}{1+a_{xy}}
\]
satisfy the condition above. Another more practically relevant case is
the log-normal distribution with suitable parameters, which corresponds
to the so-called penalty method \citep{ceperley1999penalty}.
\end{example} 
A fundamental consequence of the fact that the acceptance ratio is
approximated afresh for $x,y\in\mathsf{X}$ whenever it is needed
is that $\mathring{P}$ is a Markov chain on $(\mathsf{X},\mathcal{X})$
and has $\pi$ as invariant distribution, independently of $\{Q_{xy}\}{}_{(x,y)\in\mathsf{X}^{2}}$.
This is in contrast with pseudo-marginal algorithms, for which $\{X_{n}\}_{n\geq0}$
is not a Markov chain, but $\{X_{n},W_{n}\}_{n\geq0}$ is, and the
invariant distribution of the latter depends on $\{Q_{x}\}{}_{x\in\mathsf{X}}.$
As a result, the algorithm corresponding to $\mathring{P}$ is particularly
simple to analyse in the context of the convex ordering. Indeed if
for some $x,y\in\mathsf{X}$ we have $Q_{xy}^{(1)}\lecx Q_{xy}^{(2)}$
then from Lemma \ref{lem:cx-char} \ref{item:cx-min} we have the
following inequality for the acceptance probabilities, 
\[
\int\min\left\{ 1,r(x,y)\varpi_{2}\right\} Q_{xy}^{(2)}(\mathrm{d}\varpi_{2})\leq\int\min\left\{ 1,r(x,y)\varpi_{1}\right\} Q_{xy}^{(1)}(\mathrm{d}\varpi_{1}).
\]
In the situation where this inequality holds for any $x,y\in\mathsf{X}$,
we can apply Peskun's result stated in Theorem \ref{th:peskun} directly.
More specifically, if $\mathring{P}_{1}$ and $\mathring{P}_{2}$
are the algorithms corresponding to $\{Q_{xy}^{(1)}\}{}_{(x,y)\in\mathsf{X}^{2}}$
and $\{Q_{xy}^{(2)}\}{}_{(x,y)\in\mathsf{X}^{2}}$, the above inequality
implies that the probability of leaving any state $x\in\mathsf{X}$
is larger for $\mathring{P}_{1}$ than for $\mathring{P}_{2}$, and
therefore we conclude that $\mathrm{var}(f,\mathring{P}_{1})\le\mathrm{var}(f,\mathring{P}_{2})$
and ${\rm Gap}_{R}\bigl(\mathring{P}_{1}\bigr)\ge{\rm Gap}_{R}\bigl(\mathring{P}_{2}\bigr)$. 

One interest of identifying the convex order as an appropriate concept
for the comparison of the asymptotic properties of such algorithms
is that it allows one to use the wealth of existing results concerning
the convex order. For example it is direct to establish that the diatomic
distribution in Example \ref{ex:toyP_ring} is the worst possible
distribution among all the probability distributions with support
included in $[a_{xy}^{-1},a_{xy}]$ and the choice of $a_{xy}=1$, the ``noiseless''
algorithm, leads to the best algorithm; see Section \ref{sec:Extremal-properties}.

Turning back to the pseudo-marginal algorithms, it will be useful
in what follows to consider the expected (or at equilibrium) acceptance probability and
particularly the conditional expected acceptance probability defined
as 
\begin{align*}
\alpha(\tilde{P})&:=\int\alpha_{xy}(\tilde{P})\pi(\mathrm{d}x)q(x,\mathrm{d}y)\qquad\text{with}
\\
\alpha_{xy}\bigl(\tilde{P}\bigr)&:=\int\min\left\{
1,r(x,y)\frac{u}{w}\right\} Q_{x}(\mathrm{d}w)wQ_{y}(\mathrm{d}u),
\end{align*}
respectively. It is possible to show directly that if for some $x,y\in\mathsf{X}$
the orders $Q_{x}^{(1)}\lecx Q_{x}^{(2)}$ and $Q_{y}^{(1)}\lecx Q_{y}^{(2)}$
hold, then $\alpha_{xy}\bigl(\tilde{P}_{1}\bigr)\ge\alpha_{xy}\bigl(\tilde{P}_{2}\bigr),$
where $\tilde{P}_{1}$ and $\tilde{P}_{2}$ denote pseudo-marginal
algorithms with $\bigl\{ Q_{x}^{(1)}\bigr\}_{x\in\mathsf{X}}$ and
$\bigl\{ Q_{x}^{(2)}\bigr\}_{x\in\mathsf{X}}$ (see also Theorem \ref{thm:orderingpseudomarginals}
\ref{item:pseudo-accprob} and Theorem \ref{thm:propertiesmoon_Ps}
for a proof). If $Q_{x}^{(1)}\lecx Q_{x}^{(2)}$ for all
$x\in\mathsf{X}$, then clearly $\alpha(\tilde{P}_{1})\ge\alpha(\tilde{P_{2}})$.

It is well known that acceptance probabilities are not always a useful performance criterion 
for MH algorithms. They turn out, however, to be a relevant in the  present context. Indeed for a fixed family $\big\{q(x,\cdot)\big\}_{x \in \mathsf{X}}$ a larger value $\alpha_{xy}\bigl(\tilde{P}\bigr)$ indicates that
at stationarity the transition from $x$ to $y$ is more likely, which is somewhat reminiscent 
of the off-diagonal order in Theorem \ref{th:peskun}.
Nevertheless, as pointed out earlier, Peskun's result does not apply here since, among
other things, $\tilde{P}_{1}$ and $\tilde{P}_{2}$ do not share the
same invariant distribution. In the next section we in fact establish
that showing $Q_{x}^{(1)}\lecx Q_{x}^{(2)}$ for all $x\in\mathsf{X}$
is sufficient to imply the desired orders.


\section{Main results: Ordering pseudo-marginal MCMC}
\label{sec:pseudo-marginals-results} 

Our main results are all based on the following conditional convex
order assumption on the weight distributions. 
\begin{defn} 
\label{a:conditional-convex-order} 
Two families of weight distributions
$\{Q_{x}^{(1)}\}_{x\in\mathsf{X}}$ and $\{Q_{x}^{(2)}\}_{x\in\mathsf{X}}$
satisfy $\{Q_{x}^{(1)}\}_{x\in\mathsf{X}}\lecx \{Q_{x}^{(2)}\}_{x\in\mathsf{X}}$
if 
\[
Q_{x}^{(1)}\lecx Q_{x}^{(2)}\qquad\text{for all \ensuremath{x\in\mathsf{{X}.}}}
\]
\end{defn}
The proofs of our results are based on classical Hilbert space techniques
for the analysis of reversible Markov chains. We recall here related
definitions which will be useful throughout. Let $\mu$ be a probability
measure and $\Pi$ a $\mu$-reversible Markov transition kernel on
a measurable space $\bigl(\mathsf{E},\mathcal{F}\bigr).$ For any
probability measure $\nu$ on $(\mathsf{E,\mathcal{F})}$ and any
function $f:\mathsf{E}\rightarrow\mathbb{R}$ let, whenever the integrals
are well-defined, 
\[
\nu\bigl(f\bigr):=\int f(x)\nu({\rm d}x)\qquad\text{and}\qquad\Pi f(x):=\int f(y)\Pi\bigl(x,{\rm d}y\bigr),
\]
and for $k\geq2$, by induction, 
\[
\Pi^{k}f(x):=\int\Pi\bigl(x,{\rm d}y\bigr)\Pi^{k-1}f(y).
\]
We further denote $(\nu\Pi^{k})f:=\nu\bigl(\Pi^{k}f\bigr)$, which
can be interpreted as a probability measure. Consider next the spaces
of square integrable (and centred) functions defined respectively
as 
\begin{align*}
L^{2}(\mathsf{E},\mu)&:=\left\{
f:\mathsf{E}\rightarrow\mathbb{R}\,:\,\mu(f^{2})<\infty\right\}, \\
L_{0}^{2}(\mathsf{E},\mu)&:=\big\{ f\in L^{2}(\mathsf{E},\mu)\,:\,\mu(f)=0\big\},
\end{align*}
endowed with the inner product defined for any $f,g\in L^{2}(\mathsf{E},\mu)$
as $\left\langle f,g\right\rangle _{\mu}:=\int f(x)g(x)\mu({\rm d}x)$,
and the associated norm $\|f\|_{\mu}:=\sqrt{\smash[b]{{\langle}f,f\rangle_{\mu}}}$.
The Markov kernel $\Pi$ defines a self-adjoint operator on $L^{2}(\mathsf{E},\mu)$.
For $f\in L^{2}\bigl(\mathsf{E},\mu\bigr)$, we consider the Dirichlet
forms of $f$ associated with $\Pi$ 
\begin{align*}
\mathcal{E}_{\Pi}(f) & :=\left\langle f,(I-\Pi)f\right\rangle _{\mu}=\frac{1}{2}\int\mu(\mathrm{d}x)\Pi(x,\mathrm{d}y)\left(f(x)-f(y)\right)^{2},
\end{align*}
where $I(x,A):=\mathbb{I}\{x\in A\}$ stands for the identity operator.
The (right) spectral gap of $\Pi$ is the distance between $1$ and
the upper end of the spectrum of $\Pi$ as an operator on $L_{0}^{2}(\mathsf{E},\mu)$,
and has the following variational representation by means of the Dirichlet
form 
\[
{\rm Gap}_{R}\left(\Pi\right):=\inf\bigl\{\mathcal{E}_{\Pi}(f)\,:\, f\in L_{0}^{2}(\mathsf{E},\mu),\,\|f\|_{\mu}=1\bigr\}.
\]
Hereafter, for functions $f:\mathsf{X}\to\mathbb{R}$ we will also
denote, whenever necessary, by $f$ the functions from $\mathsf{X}\times\mathbb{R}_{+}\rightarrow\mathbb{R}$
defined by $f(x,w):=f(x)$. We now state our main result, whose proof
is postponed to Section \ref{sec:Proofs}. 
\begin{thm}
\label{thm:orderingpseudomarginals} 
Let $\pi$ be a probability distribution
on some measurable space $\bigl(\mathsf{X},\mathcal{X}\bigr)$ and
let $\tilde{P}_{1}$ and $\tilde{P}_{2}$ be two pseudo-marginal approximations
of $P$ aiming to sample from $\pi$, sharing a common family of marginal
proposal probability distributions $\{q(x,\cdot)\}_{x\in\mathsf{X}}$
but with distinct weight distributions satisfying $\{Q_{x}^{(1)}\}_{x\in\mathsf{X}}\lecx \{Q_{x}^{(2)}\}_{x\in\mathsf{X}}$.
Then, 
\begin{enumerate}[label=(\alph*)]
\item \label{item:pseudo-accprob} for any $x,y\in\mathsf{X}$, the conditional
acceptance rates satisfy $\alpha_{xy}(\tilde{P}_{1})\geq\alpha_{xy}(\tilde{P}_{2})$, 
\item \label{item:pseudo-dirichlet} for any $f:\mathsf{X}\to\mathbb{R}$,
the Dirichlet forms satisfy $\mathcal{E}_{\tilde{P}_{1}}(f)\geq\mathcal{E}_{\tilde{P}_{2}}(f)$, 
\item \label{item:pseudo-asvar} for any $f:\mathsf{X}\to\mathbb{R}$ with
$\mathrm{var}_{\pi}(f)<\infty$, the asymptotic variances satisfy
$\mathrm{var}(f,\tilde{P}_{1})\leq\mathrm{var}(f,\tilde{P}_{2})$, 
\item \label{item:pseudo-gaps} the spectral gaps satisfy ${\rm Gap}_{R}(\tilde{P}_{1})\geq\min\{{\rm Gap}_{R}(\tilde{P}_{2}),1-\tilde{\rho}_{2}^{*}\}$,
where $\tilde{\rho}_{2}^{*}:=\tilde{\pi}_{2}\text{-}\mathrm{ess\, sup}_{(x,w)}\tilde{\rho}_{2}(x,w)$,
the essential supremum of the rejection probability corresponding
to $\tilde{P}_{2}$,
\item \label{item:pseudo-gaps-continuous}if $\pi$ is not concentrated
on points, that is, $\pi(\{x\})=0$ for all $x\in\mathsf{X},$ then
\textup{${\rm Gap}_{R}(\tilde{P}_{1})\geq{\rm Gap}_{R}(\tilde{P}_{2}).$ }
\end{enumerate}
\end{thm}

\begin{rem} 
Theorem \ref{thm:orderingpseudomarginals} \ref{item:pseudo-asvar} does
not assume the finiteness of the asymptotic variances, and
accommodates the scenarios where either $\var(f,\tilde{P}_2)$ or both
$\var(f,\tilde{P}_1)$ and $\var(f,\tilde{P}_2)$ are infinite. Establishing finiteness
is a separate, but practically important, problem and we now discuss briefly what are in our view the two most applicable approaches to do so.
In the case where $\tilde{P}_2$ admits a spectral gap, then finiteness of 
$\var(f,\tilde{P}_2)$ is guaranteed, because $\pi(f^2)<\infty$ implies
that $\tilde{\pi}_2(f^2)<\infty$. In earlier work
\cite[Proposition 10]{andrieu-vihola-2012}, we have shown that if the weight
distributions $\{Q_{x}^{(2)}\}_{x\in\mathsf{X}}$ are uniformly bounded, 
and the marginal algorithm admits a spectral gap, then
so does $\tilde{P}_2$. There are also several results in the literature which 
guarantee geometric convergence and thus the existence of
spectral gaps for certain classes of Metropolis-Hastings
algorithms~\cite[e.g.][]{jarner-hansen,roberts-rosenthal-geometric}.

In the case where $\tilde{P}_2$ is sub-geometrically ergodic, and therefore does
not admit a spectral gap, then the most successful general technique
for guaranteeing the finiteness of $\var(f,\tilde{P}_2)$ consists of establishing 
sub-geometric drift and minorisation conditions
\cite{tuominen-tweedie,jarner-roberts,douc-fort-moulines-soulier}. In earlier work we have shown that such drifts hold for the pseudo-marginal kernel $\tilde{P}_2$ 
under general moment conditions on 
$\{Q_{x}^{(2)}\}_{x\in\mathsf{X}}$ when the marginal
algorithm is strongly uniformly ergodic
\cite[Proposition 30 and Corollary 31]{andrieu-vihola-2012},
and in some more specific scenarios such as the independence sampler
\cite[Corollary 27]{andrieu-vihola-2012} and the random-walk Metropolis
\cite[Theorems 38 and 45]{andrieu-vihola-2012}.
\end{rem} 

\begin{rem} 
In the context of Algorithm \ref{alg:pseudo-marginal}, if there are
two possible estimators of $\pi(Y_{n})$ that could be used, say $U_{n}^{(1)}$
and $U_{n}^{(2)}$, then $U_{n}^{(1)}\lecx U_{n}^{(2)}$ is equivalent
to our assumption $\{Q_{x}^{(1)}\}_{x\in\mathsf{X}}\lecx \{Q_{x}^{(2)}\}_{x\in\mathsf{X}}$.
This is because the convex order is preserved under scaling.
\end{rem}

As pointed out in Remark \ref{rem:cx-vs-var}, the convex order $W^{(1)}\lecx W^{(2)}$
of square-integrable random variables automatically implies ${\rm var}\bigl(W^{(1)}\bigr)\leq{\rm var}\bigl(W^{(2)}\bigr)$,
but the reverse is not true in general. A natural question is then
to ask if ${\rm var}\bigl(W_{x}^{(1)}\bigr)\leq{\rm var}\bigl(W_{x}^{(2)}\bigr)$
for all $x\in\mathsf{X}$ could be sufficient to imply, for example,
$\mathrm{var}(f,\tilde{P}_{2})\geq\mathrm{var}(f,\tilde{P}_{1})$
for $f\in L^{2}\bigl(\mathsf{X},\pi\bigr)$? The following counter-example
shows that this is not the case.
\begin{example} 
\label{ex:counterexamplevarianceorder} 
Consider the situation where
$\mathsf{X}=\{-1,1\}$, $\pi=(1/2,1/2)$, and the marginal algorithm
is a ``perfect'' independent Metropolis-Hastings (IMH) algorithm,
that is, $q(x,{\rm d}y)=\pi({\rm d}y)$, for any $x\in\mathsf{X}$.
Suppose that the weight distributions are independent of $x$ and
given by 
\[
Q({\rm d}w):=Q_{x}({\rm d}w)=\frac{b-1}{b-a}\delta_{a}(\mathrm{d}w)+\frac{1-a}{b-a}\delta_{b}(\mathrm{d}w),
\]
for some $0\leq a\le1\leq b<\infty$, and that the function of interest
is $f(x)=x$. In this case ${\rm var}_{Q}\bigl(f\bigr)=(b-1)(1-a)$
and because of the simple structure of the problem (independence with
respect to $x$ of $Q_{x}$ and the choice of an IMH) one can find
an explicit expression for the asymptotic variance 
\[
{\rm var}\bigl(f,\tilde{P}\bigr)=\frac{a(b-1)+(2b-1)b(1-a)}{b-a}.
\]
Now one can easily find numerous counterexamples such as the pairs
$(a^{(1)},b^{(1)})=(0.9208,3.0046)$ and $(a^{(2)},b^{(2)})=(0.6698,1.4620)$
for which ${\rm var}(W^{(1)})=0.1587\geq{\rm var}(W^{(2)})=0.1526$
but ${\rm var}(f,\tilde{P}_{1})=1.4577\leq{\rm var}(f,\tilde{P}_{2})=1.5632$. 
\end{example}
\begin{rem} 
Theorem \ref{thm:orderingpseudomarginals} \ref{item:pseudo-accprob}--\ref{item:pseudo-asvar}
generalise the findings in \citep{andrieu-vihola-2012} which state
similar bounds in the special case where $\tilde{P}_{1}$ corresponds
to the marginal algorithm, or equivalently, to the degenerate case
$Q_{x}^{(1)}\equiv\delta_{1}$. Note also that $\delta_{1}$ is the
unique minimal distribution in the convex order; see Section \ref{sec:Extremal-properties}.
\end{rem}

\begin{rem}
\label{rem:monotonicity-and-convexity} 
In practice one may be interested
in a sequence of estimators $\{W_{x}^{(i)}\}$, where $i\in\mathbb{N}$
is an ``accuracy parameter'' such as a number of estimators combined
by averaging. Suppose that the estimators are increasingly accurate
in the convex order, that is $W_{x}^{(i+1)}\le_{cx}W_{x}^{(i)}$,
then Theorem \ref{thm:orderingpseudomarginals} implies that the following
mappings from $\mathbb{N}$ to $\mathbb{R}_{+}$ have the monotonicity
properties
\begin{enumerate}[label=(\alph*)]
\item $i\mapsto\alpha_{xy}(\tilde{P}_{i})$ is non-decreasing, 
\item $i\mapsto\mathcal{E}_{\tilde{P}_{i}}(f)$ is non-decreasing, 
\item $i\mapsto\mathrm{var}(f,\tilde{P}_{i})$ is non-increasing. 
\end{enumerate}
\end{rem}
We suspect that in addition, in scenarios such as those of Section
\ref{sec:Averaging-in-the}, the mappings $i\mapsto\alpha_{xy}(\tilde{P}_{i})$
and $i\mapsto\mathcal{E}_{\tilde{P}_{i}}(f)$ are concave and $i\mapsto\mathrm{var}(f,\tilde{P}_{i})$
convex, but we have not yet been able to prove this conjecture. See however Proposition
\ref{prop:pseudo-convexity} for a partial result in that direction.


\section{Variational bounds for the asymptotic variance}
\label{sec:Variational-bounds-for}\label{sec:variational-peskun} 

The first result on our journey to prove Theorem \ref{thm:orderingpseudomarginals}
is a variational bound on the difference of asymptotic variances.
The result, which is of independent interest, shows that the Dirichlet
forms associated with Peskun's variance ordering result \citep{peskun,tierney-note}
need not be ordered for all functions, but only certain subclasses
of functions. We note that the result, Theorem \ref{thm:ordervariancealaTierney},
offers also a more direct proof of, for example the results in \citep{tierney-note,andrieu-vihola-2012}.
We start by stating a powerful variational representation of the quadratic
form of the inverse of a positive self-adjoint operator, attributed
to Bellman \citep{bellman1968some}, and used for example by 
Caracciolo, Pelissetto and Sokal \cite{caracciolo-pelissetto-sokal}. 
\begin{lem} 
\label{lem:variationalrepinverseoperator} 
Let $A$ be a self-adjoint
operator on a Hilbert space $\mathcal{H}$, satisfying $\left\langle f,Af\right\rangle \ge0$
for all $f\in\mathcal{H}$ and such that the inverse $A^{-1}$ exists.
Then 
\begin{equation}
\left\langle f,A^{-1}f\right\rangle =\sup_{g\in\mathcal{H}}\left[2\left\langle f,g\right\rangle -\left\langle g,Ag\right\rangle \right],
    \label{eq:bellman}
\end{equation}
 where the supremum is attained with $g=A^{-1}f$. 
\end{lem}
\noindent Proof of Lemma \ref{lem:variationalrepinverseoperator}
is given for the reader's convenience in Appendix
\ref{app:caracciolo-pelissetto-sokal}.

Lemma \ref{lem:variationalrepinverseoperator} provides us with a
quick route to Peskun type ordering. More importantly, it leads to
important quantitative bounds on the difference between the asymptotic
variances of two $\mu-$reversible Markov transition probabilities
in terms of Dirichlet forms. Suppose that $\Pi$ is a Markov kernel
on a measurable space $(\mathsf{E},\mathcal{F})$, reversible with
respect to a probability measure $\mu$, and let $\lambda\in[0,1)$
be any constant. We may introduce the self-adjoint operator (or sub-probability
kernel) $(\lambda\Pi)(x,A):=\lambda\Pi(x,A)$ and we extend the definition
of the asymptotic variance to this type of (non-probabilistic) operator
as follows. For any $f\in L^{2}(\mathsf{E},\mu)$ we let $\bar{f}:=f-\mu\bigl(f\bigr)\in L_{0}^{2}(\mathsf{E},\mu)$
and define 
\[
\mathrm{var}(f,\lambda\Pi):=\langle\bar{f},(I-\lambda\Pi)^{-1}(I+\lambda\Pi)\bar{f}\rangle_{\mu}=2\left\langle \bar{f},(I-\lambda\Pi)^{-1}\bar{f}\right\rangle _{\mu}-\|\bar{f}\|_{\mu}^{2},
\]
 where the inverse $\left(I-\lambda\Pi\right)^{-1}:=\sum_{k=0}^{\infty}\lambda^{k}\Pi^{k}$
is a well-defined bounded operator for any $\lambda\in[0,1)$. From
\citep{tierney-note} we know that 
\[
\lim_{\lambda\uparrow1}\mathrm{var}(f,\lambda\Pi)=\mathrm{var}(f,\Pi)=\lim_{M\rightarrow\infty}\mathrm{var}_{\mu}\left(M^{-1/2}{\textstyle \sum}_{i=1}^{M}f(\Phi_{i})\right),
\]
even in the case where the expression on the right hand side is infinite.
Similarly we extend the definition of Dirichlet forms to $\mathcal{E}_{\lambda\Pi}\bigl(f\bigr):=\left\langle f,(I-\lambda\Pi)f\right\rangle _{\mu}$. 
\begin{thm} 
\label{thm:ordervariancealaTierney} 
Let $\Pi_{1}$ and $\Pi_{2}$
be two Markov transition probabilities defined on some measurable
space $(\mathsf{E},\mathcal{F})$ both reversible with respect to
the probability distribution $\mu$, and let $f\in L_{0}^{2}(\mathsf{E},\mu)$.
Then, 
\begin{enumerate}[label=(\alph*)]
\item \label{item:var-bounds} for any $\lambda\in[0,1)$, 
\begin{align*}
\mathcal{E}_{\lambda\Pi_{1}}\big(\hat{f}_{1}^{\lambda}\big)-\mathcal{E}_{\lambda\Pi_{2}}\big(\hat{f}_{1}^{\lambda}\big) 
&
\leq\frac{1}{2}\Big[\mathrm{var}(f,\lambda\Pi_{2})-\mathrm{var}(f,\lambda\Pi_{1})\Big]
\\
& \le\mathcal{E}_{\lambda\Pi_{1}}\big(\hat{f}_{2}^{\lambda}\big)-\mathcal{E}_{\lambda\Pi_{2}}\big(\hat{f}_{2}^{\lambda}\big),
\end{align*}
 where $\hat{f}_{i}^{\lambda}:=\bigl(I-\lambda\Pi_{i}\bigr)^{-1}f\in L_{0}^{2}(\mathsf{E},\mu)$. 
\item \label{item:var-convexity} the function $\beta\mapsto\mathrm{var}(f,\beta\Pi_{1}+(1-\beta)\Pi_{2})$
defined for $\beta\in[0,1]$ is convex, that is, for any $\beta\in[0,1]$,
\[
\mathrm{var}\big(f,\beta\Pi_{1}+(1-\beta)\Pi_{2}\big)\leq\beta\mathrm{var}(f,\Pi_{1})+(1-\beta)\mathrm{var}(f,\Pi_{2}).
\]
\end{enumerate}
\end{thm} 
\begin{proof} 
In order to prove the results we use the variational representation
of the asymptotic variance, as suggested in \citep{caracciolo-pelissetto-sokal}
and obtained by application of Lemma \ref{lem:variationalrepinverseoperator}
for $i\in\{1,2\}$, 
\begin{align}
\mathrm{var}(f,\lambda\Pi_{i}) & =2\left\langle f,(I-\lambda\Pi_{i})^{-1}f\right\rangle _{\mu}-\|f\|_{\mu}^{2}\nonumber \\
 & =2\sup_{g\in L_{0}^{2}(\mathsf{E},\mu)}\left[2\left\langle f,g\right\rangle _{\mu}-\mathcal{E}_{\lambda\Pi_{i}}(g)\right]-\|f\|_{\mu}^{2}.
 \label{eq:var-variational}
\end{align}
Hereafter, we denote $\bar{\imath}=2$ if $i=1$ and vice versa. From
\eqref{eq:var-variational} and Lemma \ref{lem:variationalrepinverseoperator}
which states that the supremum above is attained for $\hat{f}_{i}^{\lambda}$,
we have for $i\in\{1,2\}$ 
\[
2\left[2\bigl\langle f,\hat{f}_{\bar{\imath}}^{\lambda}\bigr\rangle_{\mu}-\mathcal{E}_{\lambda\Pi_{i}}(\hat{f}_{\bar{\imath}}^{\lambda})\right]\leq\mathrm{var}(f,\lambda\Pi_{i})+\|f\|_{\mu}^{2}=2\left[2\bigl\langle f,\hat{f}_{i}^{\lambda}\bigr\rangle_{\mu}-\mathcal{E}_{\lambda\Pi_{i}}(\hat{f}_{i}^{\lambda})\right].
\]
We can now conclude \ref{item:var-bounds} by summing the above inequality
with $i=1$ and with $i=2$ multiplied by $-1$, and then dividing
by $2$. For the second item \ref{item:var-convexity}, let $\beta\in(0,1)$
and write for any $g\in L^{2}\bigl(\mathsf{E},\mu\bigr)$, 
\[
2\left\langle f,g\right\rangle _{\mu}-\mathcal{E}_{\beta\Pi_{1}+(1-\beta)\Pi_{2}}(g)=\beta\big[2\left\langle f,g\right\rangle _{\mu}-\mathcal{E}_{\Pi_{1}}(g)\big]+(1-\beta)\big[2\left\langle f,g\right\rangle _{\mu}-\mathcal{E}_{\Pi_{2}}(g)\big].
\]
The claim follows by taking the supremum over $g\in L^{2}\bigl(\mathsf{E},\mu\bigr)$,
separately for the two terms on the right hand side. 
\end{proof}


\section{Proofs by a martingale coupling of pseudo-marginal kernels}
\label{sec:Proofs} 

We preface the proof of Theorem \ref{thm:orderingpseudomarginals}
with a key result from \citep{leskela-vihola-strassen}, which ensures
that the conditional convex order implies a conditional martingale
coupling of the distributions involved.
\begin{thm}
\label{thm:conditional-martingale-coupling} 
Assume $\{Q_{x}^{(1)}\}_{x\in\mathsf{X}}\lecx \{Q_{x}^{(2)}\}_{x\in\mathsf{X}}$,
then there exists a probability kernel $(x,A)\mapsto R_{x}(A)$ from
$\mathsf{X}$ to $\mathbb{R}_{+}^{2}$ such that for any $x\in\mathsf{X}$, 
\begin{enumerate}[label=(\alph*)]
\item $R_{x}$ has marginals $Q_{x}^{(1)}$ and $Q_{x}^{(2)}$, that is,
$R_{x}(A\times\mathbb{R}_{+})=Q_{x}^{(1)}(A)$ and $R_{x}(\mathbb{R}_{+}\times A)=Q_{x}^{(2)}(A)$
for all $A\in\mathcal{B}(\mathbb{R}_{+})$,
\item \label{item:martingale-property} $R_{x}$ is the distribution of
a martingale, that is, for all $A\in\mathcal{B}(\mathbb{R}_{+})$,
$
    \int R_{x}(\mathrm{d}w\times\mathrm{d}v)v\mathbb{I}\bigl\{ w\in A\bigr\}=\int R_{x}(\mathrm{d}w\times\mathrm{d}v)w\mathbb{I}\bigl\{ w\in A\bigr\}.
$
 
\end{enumerate}
\end{thm}

\begin{rem} 
\label{rem:strassen-martingale}The property of $R_{x}$ in Lemma
\ref{thm:conditional-martingale-coupling} \ref{item:martingale-property}
holds if and only if $(W,V)\sim R_{x}(\,\cdot\,)$ satisfies $\mathbb{E}[V\mid W]=W$
almost surely. This means that $\Delta=V-W$ is a martingale difference
satisfying $\mathbb{E}[\Delta\mid W]=0$, and that $V=W+\Delta$ is
``noisier'' than $W$. 
\end{rem}

The proof of Theorem \ref{thm:conditional-martingale-coupling} given
in \citep[Theorem 1.3]{leskela-vihola-strassen} relies on the fundamental
martingale characterisation due to Strassen \citep{strassen} restated
in Theorem \ref{th:strassen}, but involves a non-trivial measurability
argument for the case where $\mathsf{X}$ is uncountable.

For the rest of this section, we assume that the conditions in Theorem
\ref{thm:orderingpseudomarginals} hold, and we denote by $\tilde{\pi}_{1}$,
$\tilde{\pi}_{2}$ the invariant distributions of $\tilde{P}_{1}$,
$\tilde{P}_{2}$, respectively. Theorem \ref{thm:conditional-martingale-coupling}
turns out to be the key instrument in the proof of Theorem \ref{thm:orderingpseudomarginals}.
It will allow us to circumvent the difficulty of having two distinct
invariant distributions $\tilde{\pi}_{1}$ and $\tilde{\pi}_{2}$
for $\tilde{P}_{1}$ and $\tilde{P}_{2}$, which is incompatible with
the Hilbert space setting. Instead, we will be working with two Markov
kernels $\breve{P}_{1}$ and $\breve{P}_{2}$ equivalent to $\tilde{P}_{1}$
and $\tilde{P}_{2}$ in a sense to be made more precise in Lemma \ref{lem:defmoonkernels}.
The kernels $\breve{P}_{1}$ and $\breve{P}_{2}$ introduced below
can be thought of as corresponding to two distinct pseudo-marginal
implementations, where $\breve{P}_{2}$ uses the ``noisiest'' approximation. 
\begin{lem} 
\label{lem:defmoonkernels} 
Let $R_{x}$ be the probability kernel
from $(\mathsf{X},\mathcal{X})$ to $(\mathbb{R}_{+}^{2},\mathcal{B}(\mathbb{R}_{+})^{2})$
from Theorem \ref{thm:conditional-martingale-coupling}. Then, the
following defines a probability distribution on $\big(\mathsf{X}\times(0,\infty)^{2},\mathcal{X}\times\mathcal{B}((0,\infty)^{2})\big)$,
\[
\breve{\pi}({\rm d}x\times{\rm d}w\times{\rm d}v):=\pi({\rm d}x)R_{x}({\mathrm{d}}w\times{\mathrm{d}}v)v,
\]
 and the following define Markov transition probabilities on $\big(\mathsf{X}\times(0,\infty)^{2},\mathcal{X}\times\mathcal{B}((0,\infty)^{2})\big)$
\begin{align*}
\breve{P}_{1}(x,w,v;{\rm d}y\times{\rm d}u\times{\rm d}t) 
 :=\;&q(x,{\rm d}y)R_{y}({\rm d}u\times{\rm d}t)\frac{t}{u}\min\left\{
1,r(x,y)\frac{u}{w}\right\} \\
 \phantom{:=}& +\delta_{x,w,v}({\rm d}y\times{\rm d}u\times{\rm d}t)\tilde{\rho}_{1}(x,w)\\
\breve{P}_{2}(x,w,v;{\mathrm{d}}y\times{\mathrm{d}}u\times{\mathrm{d}}t)
:=\;&q(x,{\rm d}y)R_{y}({\rm d}u\times{\rm d}t)\min\left\{
1,r(x,y)\frac{t}{v}\right\} \\
&+\delta_{x,w,v}({\rm d}y\times{\rm d}u\times{\rm d}t)\tilde{\rho}_{2}(x,v),
\end{align*}
with the convention $t/u=0$ for $t=u=0$, and where $\tilde{\rho}_{1}$
and $\tilde{\rho}_{2}$ stand for the rejection probabilities of $\tilde{P}_{1}$
and $\tilde{P}_{2}$, respectively. The following marginal equivalences
hold between $\breve{\pi}$, $\breve{P}_{1}$, $\breve{P}_{2}$ and
$\tilde{\pi}_{1}$, $\tilde{\pi}_{2}$, $\tilde{P}_{1}$, $\tilde{P}_{2}$,
\begin{align*}
\breve{\pi}\big({\mathrm{d}}x\times{\mathrm{d}}w\times(0,\infty)\big) & =\tilde{\pi}_{1}({\mathrm{d}}x\times{\mathrm{d}}w), \\
\breve{\pi}\big({\mathrm{d}}x\times(0,\infty)\times{\mathrm{d}}v\big)
& =\tilde{\pi}_{2}({\mathrm{d}}x\times{\mathrm{d}}v), \\
 \breve{P}_{1}\big(x,w,v;{\mathrm{d}}y\times{\mathrm{d}}u\times(0,\infty)\big)
 & =\tilde{P}_{1}(x,w;{\mathrm{d}}y\times{\mathrm{d}}u),\\
 \breve{P}_{2}\big(x,w,v;{\mathrm{d}}y\times(0,\infty)\times{\mathrm{d}}t\big) & =\tilde{P}_{2}(x,v;{\mathrm{d}}y\times{\mathrm{d}}t),
\end{align*}
where the latter two equalities hold for all $(x,w,v)\in\mathsf{X}\times(0,\infty)^{2}$.
Furthermore, both $\breve{P}_{1}$ and $\breve{P}_{2}$ are reversible
with respect to $\breve{\pi}$. 
\end{lem}
\begin{proof} 
The measure $\breve{\pi}$ is positive, and by the properties of $R_{x}$
\begin{align}
\label{eq:breve-pi-on-positive-reals}
\breve{\pi}(\mathsf{X}\times(0,\infty)^{2})&=\int_{\mathsf{X}}\pi({\mathrm{d}}x)\int_{(0,\infty)\times\mathbb{R}_{+}}R_{x}({\mathrm{d}}w\times{\mathrm{d}}v)v\\
&=\int_{\mathsf{X}}\pi({\mathrm{d}}x)\int_{(0,\infty)}Q_{x}^{(1)}({\mathrm{d}}w)w=1,\nonumber
\end{align}
and the marginal properties follow similarly. The marginal correspondence
between $\breve{P}_{1}$, $\breve{P}_{2}$ and $\tilde{P}_{1}$, $\tilde{P}_{2}$
is also immediate. Clearly $\breve{P}_{2}$ is a Metropolis-Hastings
algorithm with proposal $q(x,{\mathrm{d}}y)R_{y}({\mathrm{d}}u\times{\mathrm{d}}t)$
targeting $\breve{\pi}$, which implies also the reversibility.

We then turn to the reversibility of $\breve{P}_{1}$. We may focus
on the off-diagonal part \citep[cf.][]{tierney-note} and write for
any $A,B\in\mathcal{X}\times\mathcal{B}(\mathbb{R}_{+})^{2}$ with
$A,B\subset\{w>0,\, u>0\},$ 
\begin{align*}
 &
 \int\pi({\mathrm{d}}x)R_{x}({\mathrm{d}}w\times{\mathrm{d}}v)vq(x,{\rm
 d}y)R_{y}({\rm d}u\times{\mathrm{d}}t)\\
 &\qquad\qquad\qquad\times \frac{t}{u}\min\left\{ 1,r(x,y)\frac{u}{w}\right\} \mathbb{I}\{(x,w,v)\in A,\,(y,u,t)\in B\}\\
 & =\int\pi({\mathrm{d}}y)q(y,{\rm d}x)R_{y}({\rm
 d}u\times{\mathrm{d}}t)tR_{x}({\mathrm{d}}w\times{\mathrm{d}}v)\\
 &\qquad\qquad\qquad\times\frac{v}{w}\min\left\{ \frac{w}{u}r(y,x),1\right\} \mathbb{I}\{(y,u,t)\in B,\,(x,w,v)\in A\},
\end{align*}
which is enough to conclude. 
\end{proof}

We next introduce the spaces of square integrable functions which
are constant with respect to the last, the second last, and the two
last coordinates, respectively. 
\begin{align*}
L_{c2}^{2}(\mathsf{X}\times\mathbb{R}_{+}^{2},\breve{\pi}) & :=\bigl\{ f\,:\,\exists f_{1}\in L^{2}(\mathsf{X}\times\mathbb{R}_{+},\tilde{\pi}_{1}),\, f(x,w,v)=f_{1}(x,w)\bigr\}\\
L_{c1}^{2}(\mathsf{X}\times\mathbb{R}_{+}^{2},\breve{\pi}) & :=\bigl\{f \,:\,\exists f_{2}\in L^{2}(\mathsf{X}\times\mathbb{R}_{+},\tilde{\pi}_{2}),\, f(x,w,v)=f_{2}(x,v)\bigr\}\\
L_{c}^{2}(\mathsf{X}\times\mathbb{R}_{+}^{2},\breve{\pi}) & :=\bigl\{f\,:\,\exists\bar{f}\in L^{2}(\mathsf{X},\pi),\, f(x,w,v)=\bar{f}(x)\bigr\},
\end{align*}
where $f\in L^{2}(\mathsf{X}\times\mathbb{R}_{+}^{2},\breve{\pi})$ in
the definitions.
We denote the corresponding classes of zero-mean functions as
$L_{0,c*}^{2}(\mathsf{X}\times\mathbb{R}_{+}^{2},\breve{\pi}):=\bigl\{ f\in L_{c*}^{2}(\mathsf{X}\times\mathbb{R}_{+}^{2},\breve{\pi})\,:\,\breve{\pi}(f)=0\bigr\}$,
where `$*$' is a placeholder.
The next corollary of Lemma \ref{lem:defmoonkernels} records properties
of $\breve{P}_{i}$ on the above mentioned classes of functions.
\begin{cor}
\label{cor:correspondencePmoonPtilde} 
Let $f_{1}\in L^{2}(\mathsf{X}\times\mathbb{R}_{+},\tilde{\pi}_{1})$
and $f_{2}\in L^{2}(\mathsf{X}\times\mathbb{R}_{+},\tilde{\pi}_{2})$,
and denote $g_{1}(x,w,v):=f_{1}(x,w)\in L_{c2}^{2}(\mathsf{X}\times\mathbb{R}_{+}^{2},\breve{\pi})$
and $g_{2}(x,w,v):=f_{2}(x,v)\in L_{c1}^{2}(\mathsf{X}\times\mathbb{R}_{+}^{2},\breve{\pi})$.
Then we have the correspondence for $k\geq1$ 
\begin{align}
\breve{\pi}(g_{1}) & =\tilde{\pi}_{1}(f_{1}), & \breve{P}_{1}^{k}g_{1}(x,w,v) & =\tilde{P}_{1}^{k}f_{1}(x,w)\in L_{c2}^{2}(\mathsf{X}\times\mathbb{R}_{+}^{2},\breve{\pi})\label{eq:moon-vs-pseudo1}\\
\breve{\pi}(g_{2}) & =\tilde{\pi}_{2}(f_{2}), & \breve{P}_{2}^{k}g_{2}(x,w,v) & =\tilde{P}_{2}^{k}f_{2}(x,v)\in L_{c1}^{2}(\mathsf{X}\times\mathbb{R}_{+}^{2},\breve{\pi}).\label{eq:moon-vs-pseudo2}
\end{align}
and as a result for $k\geq1$ the following identities hold
\begin{align*}
\bigl\langle g_{1},\breve{P}_{1}^{k}g_{1}\bigr\rangle_{\breve{\pi}}
&=\bigl\langle
f_{1},\tilde{P}_{1}^{k}f_{1}\bigr\rangle_{\tilde{\pi}_{1}},
& \bigl\langle g_{2},\breve{P}_{2}^{k}g_{2}\bigr\rangle_{\breve{\pi}}
&=\bigl\langle
f_{2},\tilde{P}_{2}^{k}f_{2}\bigr\rangle_{\tilde{\pi}_{2}},\\
\var(g_{1},\breve{P}_{1})&=\var(f_{1},\tilde{P}_{1}),
& \var(g_{2},\breve{P}_{2})&=\var(f_{2},\tilde{P}_{2}).
\end{align*}
\end{cor}
\begin{proof} 
All the claims are direct, except the last two. Recall that (e.g.
\citep{tierney-note}) for a generic Markov kernel $\Pi$ reversible
with respect to a probability distribution $\mu$ and $\psi\in L_{0}^{2}\bigl(\mathsf{E},\mu\bigr)$
we have 
\begin{align*}
\var(\psi,\Pi) & 
=\lim_{M\to\infty}\frac{1}{M}\bigg(\sum_{i=1}^{M}\E_{\mu}[\psi^{2}(\Phi_{i})]+2\sum_{i=1}^{M}\sum_{j=i+1}^{M}\E_{\mu}[\psi(\Phi_{i})\psi(\Phi_{j})]\bigg)\\
 & =\mu(\psi^{2})+\lim_{n\to\infty}\frac{2}{M}\sum_{i=1}^{M}\sum_{j=i+1}^{M}\bigl\langle\psi,\Pi^{j-i}\psi\bigr\rangle_{\mu}\quad.
\end{align*}
The last two claims now follow from the correlation equivalences.
\end{proof}

Note that by Corollary \ref{cor:correspondencePmoonPtilde}, $L_{c}^{2}(\mathsf{X}\times\mathbb{R}_{+}^{2},\breve{\pi})\subset L_{c1}^{2}(\mathsf{X}\times\mathbb{R}_{+}^{2},\breve{\pi})\cap L_{c2}^{2}(\mathsf{X}\times\mathbb{R}_{+}^{2},\breve{\pi})\subset L^{2}(\mathsf{X}\times\mathbb{R}_{+}^{2},\breve{\pi})$
and the same inclusions hold with centred versions $L_{0,c*}(\mathsf{X}\times\mathbb{R}_{+}^{2},\breve{\pi})$.

We now state the key result which relates various quantities related
to the pseudo-marginal algorithms $\tilde{P}_{1}$ and $\tilde{P}_{2}$
and their counterparts $\breve{P}_{1}$ and $\breve{P}_{2}.$
\begin{thm}
\label{thm:propertiesmoon_Ps} 
Let $\breve{\pi}$, $\breve{P}_{1}$
and $\breve{P}_{2}$ be as defined in Lemma \ref{lem:defmoonkernels}.
Then, 
\begin{enumerate}[label=(\alph*)]
\item \label{item:moon-accprob} $\alpha_{xy}\bigl(\tilde{P}_{1}\bigr)\geq\alpha_{xy}\bigl(\tilde{P}_{2}\bigr)$
for any $x,y\in\mathsf{X}$,
\item \label{item:moon-dirichlet} $\mathcal{E}_{\breve{P}_{1}}(g)\geq\mathcal{E}_{\breve{P}_{2}}(g)$
for any $g\in L_{c2}^{2}(\mathsf{X}\times\mathbb{R}_{+}^{2},\breve{\pi})$, 
\item \label{item:moon-asvar} $\mathrm{var}(f,\breve{P}_{1})\le\mathrm{var}(f,\breve{P}_{2})$
for any $f\in L_{c2}^{2}(\mathsf{X}\times\mathbb{R}_{+}^{2},\breve{\pi})$, 
\item \label{item:ordered-spectral-gaps} 
  ${\rm Gap}_{R}\bigl(\tilde{P}_{1}\bigr)\geq{\rm Gap}_{R}\bigl(\breve{P}_{2}\bigr)$.
\end{enumerate}
\end{thm}
\begin{proof} 
We first consider \ref{item:moon-accprob} and \ref{item:moon-dirichlet}.
Fix $x,y\in\mathsf{X}$. Then for any bounded function $h:(\mathsf{X}\times\mathbb{R}_{+})^{2}\to\mathbb{R}_{+}$
by the properties of $R_{x}$ and $R_{y}$ and by Jensen's inequality,
\begin{align*}
\int R_{x}({\mathrm{d}}w&\times{\mathrm{d}}v)vR_{y}({\mathrm{d}}u\times{\mathrm{d}}t)  \frac{t}{u}\min\Big\{1,r(x,y)\frac{u}{w}\Big\} h(x,w,y,u)\\
 & =\int R_{x}({\mathrm{d}}w\times{\mathrm{d}}v)R_{y}({\mathrm{d}}u\times{\mathrm{d}}t)\min\big\{ w,r(x,y)u\big\} h(x,w,y,u)\\
 & \ge\int R_{x}({\mathrm{d}}w\times{\mathrm{d}}v)R_{y}({\mathrm{d}}u\times{\mathrm{d}}t)\min\big\{ v,r(x,y)t\big\} h(x,w,y,u)\\
 & =\int R_{x}({\mathrm{d}}w\times{\mathrm{d}}v)vR_{y}({\mathrm{d}}u\times{\mathrm{d}}t)\min\Big\{1,r(x,y)\frac{t}{v}\Big\} h(x,w,y,u).
\end{align*}
We deduce \ref{item:moon-accprob} with $h\equiv1$ and by using the
correspondence established in Lemma \ref{lem:defmoonkernels}. We
also have, with functions such that $h(x,w,x,w)=0$ for all $(x,w)\in\mathsf{X}\times(0,\infty)$,
\begin{align*}
\int\breve{\pi}({\mathrm{d}}x\times{\mathrm{d}}w&\times{\mathrm{d}}v)  \breve{P}_{1}(x,w,v;{\mathrm{d}}y\times{\mathrm{d}}u\times{\mathrm{d}}t)h(x,w,y,u)\\
 & \ge\int\breve{\pi}({\mathrm{d}}x\times{\mathrm{d}}w\times{\mathrm{d}}v)\breve{P}_{2}(x,w,v;{\mathrm{d}}y\times{\mathrm{d}}u\times{\mathrm{d}}t)h(x,w,y,u).
\end{align*}
Claim \ref{item:moon-dirichlet} is now obtained by letting $h(x,w,y,u)=\min\{m,\big(g(x,w)-g(y,u)\big)^{2}\}$
and by monotone convergence as $m\to\infty$.

In \ref{item:moon-asvar}, we may assume without loss of generality
that $f\in L_{0,c2}^{2}(\mathsf{X}\times\mathbb{R}_{+}^{2},\breve{\pi})$.
For any $\lambda\in[0,1)$, note that by Corollary \ref{cor:correspondencePmoonPtilde}
$\hat{f}_{1}^{\lambda}:=\bigl(I-\lambda\breve{P}_{1}\bigr)^{-1}f=\sum_{k=0}^{\infty}\lambda^{k}(\breve{P}_{1})^{k}f$
satisfies $\hat{f}_{1}^{\lambda}(x,w,v)=\hat{f}_{1}^{\lambda}(x,w)\in L_{0,c2}^{2}(\mathsf{X}\times\mathbb{R}_{+}^{2},\breve{\pi})$.
We may now apply \ref{item:moon-dirichlet} and Theorem \ref{thm:ordervariancealaTierney}
to deduce that 
\[
0\le2\lambda\big[\mathcal{E}_{\breve{P}_{1}}(\hat{f}_{1}^{\lambda})-\mathcal{E}_{\breve{P}_{2}}(\hat{f}_{1}^{\lambda})\big]=2\big[\mathcal{E}_{\lambda\breve{P}_{1}}(\hat{f}_{1}^{\lambda})-\mathcal{E}_{\lambda\breve{P}_{2}}(\hat{f}_{1}^{\lambda})\big]\leq\mathrm{var}(f,\lambda\breve{P}_{2})-\mathrm{var}(f,\lambda\breve{P}_{1}).
\]
 If $\mathrm{var}(f,\breve{P}_{2})$ is infinite, then the claim holds
trivially. Suppose now that $\mathrm{var}(f,\breve{P}_{2})$ is finite,
then taking the limit $\lambda\uparrow1$ ensures that $\mathrm{var}(f,\breve{P}_{2})\ge\mathrm{var}(f,\breve{P}_{1})$.

For \ref{item:ordered-spectral-gaps}, by the variational definition
of the right spectral gap there exists a sequence of functions $\psi_{i}\in L_{0}^{2}(\mathsf{X}\times\mathbb{R}_{+},\tilde{\pi}_{1})$
with $\mathrm{var}_{\tilde{\pi}_{1}}(\psi_{i})=1$ such that 
\[
\lim_{i\rightarrow\infty}\mathcal{E}_{\tilde{P}_{1}}(\psi_{i})={\rm Gap}_{R}\bigl(\tilde{P}_{1}\bigr).
\]
Notice that denoting $\psi_{i}(x,w,v):=\psi_{i}(x,w)$ Corollary \ref{cor:correspondencePmoonPtilde}
implies $\mathrm{var}_{\breve{\pi}}(\psi_{i})=\mathrm{var}_{\tilde{\pi}_{1}}(\psi_{i})=1$
and $\breve{P}_{1}\psi_{i}(x,w,v)=\tilde{P}_{1}\psi_{i}(x,w)$, and
therefore $\mathcal{E}_{\breve{P}_{1}}(\psi_{i})=\mathcal{E}_{\tilde{P}_{1}}(\psi_{i})$.
Now, \ref{item:moon-dirichlet} allows us to conclude that 
\[
{\rm Gap}_{R}\bigl(\breve{P}_{2}\bigr)\leq\liminf_{i\to\infty}\mathcal{E}_{\breve{P}_{2}}(\psi_{i})\leq\liminf_{i\to\infty}\mathcal{E}_{\breve{P}_{1}}(\psi_{i})=\lim_{i\to\infty}\mathcal{E}_{\tilde{P}_{1}}(\psi_{i}).
\qedhere
\]
\end{proof}

We are now ready to conclude the proof of our main result.
\begin{proof}[Proof of Theorem \ref{thm:orderingpseudomarginals}] 
 \label{proof:oftheorem2} The acceptance rate order \ref{item:pseudo-accprob}
is proved in Theorem \ref{thm:propertiesmoon_Ps} \ref{item:moon-accprob}.
For what follows, let $K_{1}(x,u;\,\cdot\,)$ and $K_{2}(x,w;\,\cdot\,)$
be (regular) conditional distributions such that 
\begin{align*}
R_{x}({\rm d}w\times{\rm d}u)&=Q_{x}^{(1)}({\rm d}w)K_{2}(x,w;{\rm
d}u),\\
R_{x}({\rm d}w\times{\rm d}u)&=Q_{x}^{(2)}({\rm d}u)K_{1}(x,u;{\rm d}w),
\end{align*}
 and define the following sub-probability kernels corresponding to
the acceptance parts of $\tilde{P}_{1}$ and $\tilde{P}_{2}$, 
\begin{align*}
\tilde{p}_{1}(x,w;{\rm d}y\times{\rm d}u) & :=q(x,{\rm
d}y)Q_{y}^{(1)}({\rm d}u)\min\left\{ 1,r(x,y)\frac{u}{w}\right\}, \\
\tilde{p}_{2}(x,v;{\rm d}y\times{\rm d}t) & :=q(x,{\rm d}y)Q_{y}^{(2)}({\rm d}t)\min\left\{ 1,r(x,y)\frac{t}{v}\right\} .
\end{align*}
With these, we may write 
\begin{align*}
\breve{P}_{1}(x,w,v;{\rm d}y\times{\rm d}u\times{\rm d}t)
=\;&\tilde{p}_{1}(x,w;{\rm d}y\times{\rm d}u)K_{2}(y,u;{\rm
d}t)\frac{t}{u}\\
&+\delta_{x,w,v}({\rm d}y\times{\rm d}u\times{\rm d}t)\tilde{\rho}_{1}(x,w)\\
\breve{P}_{2}(x,w,v;{\rm d}y\times{\rm d}u\times{\rm d}t)
=\;&\tilde{p}_{2}(x,v;{\rm d}y\times{\rm d}t)K_{1}(y,t;{\rm d}u)\\
&+\delta_{x,w,v}({\rm d}y\times{\rm d}u\times{\rm d}t)\tilde{\rho}_{2}(x,v).
\end{align*}
In both situations, we are now in the setting of Lemma \ref{lem:augmented-gaps}
in Appendix \ref{sec:augmented} with $\mathsf{E}=\mathsf{X}\times(0,\infty)$
and $\mathsf{S}=(0,\infty)$.

Define $h(x,w,v):=f(x)$ and $g(x,w):=f(x)$, then we have from Theorem
\ref{thm:propertiesmoon_Ps} \ref{item:moon-dirichlet} that $\mathcal{E}_{\tilde{P}_{1}}(g)=\mathcal{E}_{\breve{P}_{1}}(h)\geq\mathcal{E}_{\breve{P}_{2}}(h)=\mathcal{E}_{\tilde{P}_{2}}(g)$,
which concludes the proof of \ref{item:pseudo-dirichlet}. Claim \ref{item:pseudo-asvar}
follows from Theorem \ref{thm:propertiesmoon_Ps} \ref{item:moon-asvar}
and Corollary \ref{cor:correspondencePmoonPtilde}. Now recall that
Theorem \ref{thm:propertiesmoon_Ps} \ref{item:ordered-spectral-gaps}
states that ${\rm Gap}_{R}(\tilde{P}_{1})\ge{\rm Gap}_{R}(\breve{P}_{2})$,
and Lemma \ref{lem:augmented-gaps} reads ${\rm Gap}_{R}(\breve{P}_{2})\ge\min\{{\rm Gap}_{R}(\tilde{P}_{2}),1-\tilde{\rho}_{2}^{*}\}$,
which concludes the proof of \ref{item:pseudo-gaps}. Finally, \ref{item:pseudo-gaps-continuous}
is a consequence of Remark \ref{rem:continuous-gaps} in Appendix
\ref{sec:augmented}.
\end{proof}

We conclude this section by a partial result concerning the convexity
and concavity of the expected acceptance rates, the Dirichlet forms
and the asymptotic variances as discussed in Remark \ref{rem:monotonicity-and-convexity}.
We start by stating simple extensions of Theorem \ref{thm:conditional-martingale-coupling}
and Lemma \ref{lem:defmoonkernels} to the case of an arbitrary number
of distributions, which may be useful also in other contexts. Note
that here the indices are reversed in comparison to Remark \ref{rem:monotonicity-and-convexity};
$Q_{x}^{(1)}$ corresponds to the least noisy estimate.
\begin{lem}
\label{lem:strassen-existence-n} 
Suppose that $\{Q_{x}^{(1)}\}_{x\in\mathsf{X}}\le_{cx}\cdots\le_{cx}\{Q_{x}^{(n)}\}_{x\in\mathsf{X}}$,
then there exists a kernel $(x,A)\mapsto R_{x}(A)$ from $\mathsf{X}$
to $\mathbb{R}_{+}^{n}$ such that, with notation $w_{1:i}:=(w_{1},\ldots,w_{i})$,
\begin{enumerate}[label=(\alph*)]
\item For all $i=1,\ldots,n$, 
  $R_{x}(\mathbb{R}_{+}^{i-1}\times A\times\mathbb{R}_{+}^{n-i-1})=Q_{x}^{(i)}(A)$
\item For all $A\in\mathcal{B}(\mathbb{R}_{+}^{i-1})$ and all $i=2,\ldots,n$,
\[
    \textstyle\int R_{x}({\rm d}w_{1:n})w_{i}\mathbb{I}\{w_{1:i-1}\in A\}
    =\int R_{x}({\rm d}w_{1:n})w_{i-1}\mathbb{I}\{w_{1:i-1}\in A\}.
\]
\item For all $i=3,\ldots,n$, all $A_{1},\ldots,A_{i-1}\in\mathcal{B}(\mathbb{R}_{+})$
and all bounded measurable $f:\mathbb{R}_{+}\to\mathbb{R}$,
\begin{multline*}
    \textstyle\int
    R_{x}(\mathrm{d}w_{1:n})f(w_{i})\mathbb{I}\{w_{1}\in
      A_{1},\ldots,w_{i-1}\in A_{i-1}\} \\
    =\textstyle\int R_{x}(\mathbb{d}w_{1:n})f(w_{i})\mathbb{I}\{w_{i-1}\in A_{i-1}\}.
\end{multline*}
\end{enumerate}
\end{lem}
\begin{proof} 
For the existence of $R_{x}$, consider Theorem \ref{thm:conditional-martingale-coupling}
applied to each pair $\{Q_{x}^{(i-1)}\}_{x\in\mathsf{X}}\le_{cx}\{Q_{x}^{(i)}\}_{x\in\mathsf{X}}$,
resulting in $R_{x}^{(i)}({\rm d}w_{i-1}\times{\rm d}w_{i})$, and
let $K^{(i)}(x,w_{i-1};{\rm d}w_{i})$ be (regular) conditional probabilities
such that $R_{x}^{(i)}({\rm d}w_{i-1}\times{\rm d}w_{i})=Q_{x}^{(i-1)}({\rm d}w_{i-1})K(x,w_{i-1};{\rm d}w_{i})$.
The claim holds for $n=2$ because then $R_{x}=R_{x}^{(2)}$. For
$n\ge3$, assuming that the claim holds with $n-1,$ it is straightforward
to check that the extension 
\[
R_{x}({\rm d}w_{1:n})=R_{x}({\rm d}w_{1:{n-1}})K^{(n)}(x,w_{1:n-1};{\rm d}w_{n})
\]
satisfies the required properties. 
\end{proof}

\begin{rem} 
If $\hat{W}{}_{1:n}\sim R_{x},$ then $\hat{W}_{1:n}$ is a Markovian
martingale, that is, $\mathbb{E}[f(\hat{W}_{i})\mid\hat{W}_{1:i-1}]=\mathbb{E}[\hat{W}_{i}\mid\hat{W}_{i-1}]$
and $\mathbb{E}\big[\hat{W}_{i}\,\big|\,\hat{W}_{1:i-1}\big]=\hat{W}_{i-1}$
for $i=2,\ldots,n$ and all bounded measurable $f:\mathbb{R}_{+}\to\mathbb{R}.$
\end{rem}

\begin{lem} 
\label{lem:extended-pmoon} 
Let $\{Q_{x}^{(1)}\}_{x\in\mathsf{X}}\le_{cx}\cdots\le_{cx}\{Q_{x}^{(n)}\}_{x\in\mathsf{X}}$
and let $R_{x}$ be as in Lemma \ref{lem:strassen-existence-n}. Define
the probability distribution 
\[
\breve{\pi}({\rm d}x\times{\rm d}w_{1:n}):=\pi({\rm d}x)R_{x}({\mathrm{d}}w_{1:n})w_{n},
\]
 and the following Markov transition probabilities on $\big(\mathsf{X}\times(0,\infty)^{n},\mathcal{X}\times\mathcal{B}((0,\infty)^{n})\big)$
\begin{align*}
\breve{P}_{i}(x,w_{1:n};{\rm d}y\times{\rm d}u_{1:n})  :=\;&q(x,{\rm
d}y)R_{y}({\rm d}u_{1:n})\frac{u_{n}}{u_{i}}\min\left\{
1,r(x,y)\frac{u_{i}}{w_{i}}\right\} \\
&+\delta_{x,w_{1:n}}({\rm d}y\times{\rm d}u_{1:n})\tilde{\rho}_{i}(x,w_{i}).
\end{align*}
Then the $\breve{P}_{i}$s are reversible with respect to $\breve{\pi}$,
and satisfy the marginal correspondence
\begin{align*}
\breve{\pi}(\mathrm{d}x\times\mathbb{R}_{+}^{i-1}\times\mathrm{d}w_{i}\times\mathbb{R}_{+}^{n-i-1}) & =\tilde{\pi}_{i}(\mathrm{d}x\times\mathrm{d}w_{i})\\
\breve{P}_{i}(x,w_{1:n};\mathrm{d}y\times\mathbb{R}_{+}^{i-1}\times\mathrm{d}u_{i}\times\mathbb{R}_{+}^{n-i-1}) & =\tilde{P}{}_{i}(x,w_{i};\mathrm{d}y\times\mathrm{d}u_{i}).
\end{align*}
\end{lem}

\noindent The proof is similar to Lemma \ref{lem:defmoonkernels}. 

We now give our partial result relying on an abstract condition on
the Dirichlet forms of the augmented kernels $\breve{P}_{i}$ defined
in Lemma \ref{lem:extended-pmoon}.
\begin{prop}
\label{prop:pseudo-convexity}Let $\{Q_{x}^{(1)}\}_{x\in\mathsf{X}}\le_{cx}\cdots\le_{cx}\{Q_{x}^{(n)}\}_{x\in\mathsf{X}}$
and let $\breve{P}_{i}$ be as defined in Lemma \ref{lem:extended-pmoon}.
If for all $i=2,\ldots,n$ and any $g_{i}\in L_{0}^{2}(\mathsf{X}\times\mathbb{R}_{+}^{2},\breve{\pi})$
such that $g_{i}(x,w_{1:n})=h(x,w_{i})\in L^{2}(\tilde{\pi}_{i},\mathsf{X}\times\mathbb{R}_{+})$
it holds that 
\[
\mathcal{E}_{\breve{P}_{i-1}}(g_{i})-\mathcal{E}_{\breve{P}_{i}}(g_{i})\le\mathcal{E}_{\breve{P}_{i}}(g_{i})-\mathcal{E}_{\breve{P}_{i+1}}(g_{i}),
\]
 then for any function $f\in L^{2}(\mathsf{X},\pi)$, 
\[
\mathrm{var}(f,\tilde{P}_{i})-\mathrm{var}(f,\tilde{P}_{i-1})\le\mathrm{var}(f,\tilde{P}_{i+1})-\mathrm{var}(f,\tilde{P}_{i}),
\]
 whenever the quantities above are finite.\end{prop}
\begin{proof}
Without loss of generality, we may assume $f\in L_{0}^{2}(\mathsf{X},\pi)$,
so that $\hat{g}_{i}^{\lambda}:=\sum_{k\ge0}\lambda\breve{P}_{i}f\in L_{0}^{2}(\mathsf{X}\times\mathbb{R}_{+}^{n},\breve{\pi})$,
and $\hat{g}_{i}^{\lambda}(x,w_{1:n})$ depends only on $x$ and $w_{i}$.
By Theorem \ref{thm:ordervariancealaTierney}, 
\begin{align*}
\big(\mathrm{var}(g,\lambda\breve{P}_{i})-\mathrm{var}(g,\lambda\breve{P}_{i-1})\big)&\leq2\big[\mathcal{E}_{\lambda\breve{P}_{i-1}}\big(\hat{g}_{i}^{\lambda}\big)-\mathcal{E}_{\lambda\breve{P}_{i}}\big(\hat{g}_{i}^{\lambda}\big)\big]\\
& =2\lambda\big[\mathcal{E}_{\breve{P}_{i-1}}\big(\hat{g}_{i}^{\lambda}\big)-\mathcal{E}_{\breve{P}_{i}}\big(\hat{g}_{i}^{\lambda}\big)\big],
\end{align*}
 and similarly 
\[
2\lambda\big[\mathcal{E}_{\breve{P}_{i}}\big(\hat{g}_{i}^{\lambda}\big)-\mathcal{E}_{\breve{P}_{i+1}}\big(\hat{g}_{i}^{\lambda}\big)\big]\leq\mathrm{var}(g,\lambda\breve{P}_{i+1})-\mathrm{var}(g,\lambda\breve{P}_{i}).
\]
 Because $\mathcal{E}_{\breve{P}_{i-1}}\big(\hat{g}_{i}^{\lambda}\big)-\mathcal{E}_{\breve{P}_{i}}\big(\hat{g}_{i}^{\lambda}\big)\le\mathcal{E}_{\breve{P}_{i}}\big(\hat{g}_{i}^{\lambda}\big)-\mathcal{E}_{\breve{P}_{i+1}}\big(\hat{g}_{i}^{\lambda}\big)$,
we obtain the desired variance bound for $\breve{P}_{i-1}$, $\breve{P}_{i}$
and $\breve{P}_{i+1}$. Because the variances are equal to those of
$\tilde{P}_{i-1}$, $\tilde{P}_{i}$ and $\tilde{P}_{i+1}$ as observed
in the proof of Theorem \ref{thm:orderingpseudomarginals}, we conclude
the proof.
\end{proof}


\section{Applications}
\label{sec:Applications} 

The convex order is a well researched topic with a rich and extensive
literature where numerous properties have been established for various
purposes; see for example \citep{mullercomparison,shaked-shanthikumar}
for recent booklength overviews. For example the convex order is closed
under linear combinations and numerous parametric families of distributions
can be convex ordered in terms of their parameters. Conditioning improves
on convex order, that is, $\mathbb{E}\bigl[W\mid Z\bigr]\lecx W$
for some random variable $Z$ \citep[see Theorem 3.A.20. for a more general scenario, in][]{shaked-shanthikumar},
therefore establishing that, as expected, ``Rao-Blackwellisation''
is beneficial in the present context. We detail here applications
of such properties directly relevant to the pseudo-marginal context. 

We first show in Section \ref{sec:Averaging-in-the} that the theory
of majorisation provides us with a tool to compare algorithms when
averaging a number of independent realisations of an approximation.
As a by-product we establish that increasing the number of copies
always improves performance. While this result is not entirely surprising,
Example \ref{ex:counterexamplevarianceorder} is a reminder that the
behaviour of these algorithms can be counter-intuitive and surprising.
In addition, establishing the result directly seems to be far from
obvious while it follows here directly from the convex order. Our
second application is more interesting in terms of methodology. It
is concerned with stratification, which is often easy to implement
without additional computational cost. In particular, we establish
in Section \ref{sub:Stratification} that the standard application
of this variance reduction approach to approximate Bayesian computation
(ABC) MCMC always improves performance in this context. We conclude
by considering extremal distributions in Section \ref{sec:Extremal-properties}
and discuss what information they provide on the efficiency under
certain constraints. We point again to a recent application of our
work in \citep{2014arXiv1404.6298B} in order to establish quantitative
bounds.


\subsection{Averaging and performance monotonicity in the pseudo-marginal algorithm}
\label{sec:Averaging-in-the} 

A simple and practical way to reduce variability of an estimator is
to average multiple independent realisations of this estimator--this
is a particularly interesting and relevant strategy given the advent
of cheap and widely available parallel computing architectures; see
\cite{drovandi2014} for a recent application of this idea to pseudo-marginal
algorithms. It is standard to show that for $N$ independent and identical
realisations of an estimator the choice of uniform weights $1/N$
is optimum in terms of variance when linear combinations are considered.
It is then a consequence that such equal weight averaging reduces
the variance in a monotonic fashion as the number of copies increases.
One may wonder whether averaging always improves performance of a
pseudo-marginal algorithm, especially in the light of Example \ref{ex:counterexamplevarianceorder}
where we have showed that the variance is not a reliable criterion
in this context. As we shall see, however, the answer to this question
is positive, and a direct consequence of the convex order. In fact,
we are able to prove this result in a slightly more general scenario
where the copies are only assumed to be exchangeable.

We preface our result with some background. Assume $Z(1),Z(2),\ldots,Z(N)$
are exchangeable and non-negative random variables of unit expectation
and denote $Z:=\big(Z(1),Z(2),\ldots,Z(N)\big)$. We introduce the
simplex $\mathcal{S}_{N}:=\big\{\lambda:=\big(\lambda(1),\lambda(2),\ldots,\lambda(N)\big)\in[0,1]^{N}\,:\,\sum_{i=1}^{N}\lambda(i)=1\big\}$.
We consider below convex combinations of the elements of $Z$ in terms
of weights in $\mathcal{S}_{N}$ and to that purpose will use for
$a,b\in\mathbb{R}^{N}$ the notation $(a,b):=\sum_{i=1}^{N}a(i)b(i)$.
We will also denote the components of any $a\in\mathbb{R^{N}}$ in
decreasing order as $\max_{i}a(i)=a[1]\ge a[2]\ge\cdots\ge a[N]=\min_{i}a(i)$.
We introduce next the notions of Schur concavity and majorisation
\citep{marshall2010inequalities}.
\begin{defn}[Majorisation and Schur-concavity]
\label{def:majorizationschurconcavity}Suppose that $\lambda,\mu\in\mathbb{R}^{N}$.
\begin{enumerate}[label=(\alph*)]
\item We say that $\mu$ majorises $\lambda$, denoted $\lambda\prec\mu$,
if $\sum_{i=1}^{k}\lambda[i]\leq\sum_{i=1}^{k}\mu[i]$ for all $k=1,\ldots,N$, 
\item A function $\phi:\mathbb{R}^{N}\rightarrow\mathbb{R}$ is said to
be Schur concave if $\lambda\prec\mu$ implies $\phi(\lambda)\ge\phi(\mu)$,
and $\phi$ is Schur convex if $\lambda\prec\mu$ implies $\phi(\lambda)\le\phi(\mu)$.
\end{enumerate}
\end{defn}
We state next a well-known result which establishes that convex combinations
of exchangeable random variables with majorised weights imply a convex
order on convex linear combinations.
\begin{thm} 
\label{thm:majorisation} 
For any $\lambda,\mu\in\mathcal{S}_{N}$
such that $\lambda\prec\mu$, we have 
\[
\bigl(\lambda,Z\bigr)\lecx \bigl(\mu,Z\bigr).
\]
\end{thm}

\noindent For a proof, see for example \citep[Theorem 3.A.35]{shaked-shanthikumar}.

Let then $Z_{x}=\big(Z_{x}(1),Z_{x}(2),\ldots,Z_{x}(N)\big)$ for
any $x\in\mathsf{X}$ stand for an exchangeable vector as above, and
let $\lambda$,$\mu\in\mathcal{S}_{N}$. Consider the weights $W_{x}^{(\lambda)}:=(\lambda,Z_{x})$
and $W_{x}^{(\mu)}:=(\mu,Z_{x})$, which are non-negative and have
unit expectation, and let $\tilde{P}_{\lambda}$ and $\tilde{P}_{\mu}$
denote the pseudo-marginal algorithms corresponding to $\{W_{x}^{(\lambda)}\}_{x\in\mathsf{X}}$
and $\{W_{x}^{(\mu)}\}_{x\in\mathsf{X}}$, respectively.
\begin{thm}
\label{thm:schur-pseudo} 
Assume that $\lambda,\mu\in\mathcal{S}_{N}$
satisfy $\lambda\prec\mu$. Then, for any $x,y\in\mathsf{X}$ and
any $f\in L^{2}(\mathsf{X},\pi)$,
\[
\alpha_{xy}\bigl(\tilde{P}_{\lambda}\bigr)\geq\alpha_{xy}\bigl(\tilde{P}_{\mu}\bigr)\qquad\text{and}\qquad\mathrm{var}\bigl(f,\tilde{P}_{\lambda}\bigr)\le\mathrm{var}\bigl(f,\tilde{P}_{\mu}\bigr),
\]
that is, the expected acceptance probability is Schur concave, while
the asymptotic variance is Schur convex. 

Furthermore, if $\pi$ is not concentrated on points, \textup{${\rm Gap}_{R}(\tilde{P}_{\lambda})\ge{\rm Gap}_{R}(\tilde{P}_{\mu})$,
that is, the right spectral gap is Schur concave.}
\end{thm} 
\begin{proof} 
Follows directly from Theorems \ref{thm:orderingpseudomarginals}
and \ref{thm:majorisation}. 
\end{proof} 

\begin{rem} 
It is clear that Theorem \ref{thm:schur-pseudo} can be generalised
to incorporate state dependent weights, $\lambda=\{\lambda_{x}\}_{x\in\mathsf{X}}$
and $\mu=\{\mu_{x}\}_{x\in\mathsf{X}}$ where $\lambda_{x}$,$\mu_{x}\in\mathcal{S}_{N}$
and use $W_{x}^{(\lambda)}:=(\lambda_{x},Z_{x})$ and $W_{x}^{(\lambda)}:=(\lambda_{x},Z_{x})$.
The result also generalises to infinite exchangeable sequences $Z_{x}=\big(Z_{x}(1),Z_{x}(2),\ldots\big)$
and $\lambda_{x},\mu_{x}\in\mathcal{S}_{\infty}:=\big\{\lambda\in[0,1]^{\infty}\,:\,\sum_{i=1}^{\infty}\lambda(i)=1\big\}$
and letting $W_{x}^{(\lambda)}:=(\lambda_{x},Z_{x}):=\sum_{i=1}^{\infty}\lambda_{x}(i)Z_{x}(i)$.
\end{rem} 
For any $k\in\{1,\ldots,N\}$ we define $u_{k}\in\mathcal{S}_{N}$
as the uniform weights $u_{k}:=(1/k,\ldots,1/k,0,\ldots,0)$, that
is, the first $k$ components are non-zero and are all equal. The
next result shows that the optimal weighting of $N$ estimators is
the uniform weighting, and that every extra sample improves performance.
\begin{cor} 
\label{cor:optimal-weights} 
For any $\lambda\in\mathcal{S}_{N}$ all
$x,y\in\mathsf{X}$ and $f\in L^{2}(\mathsf{X},\pi)$,
\[
\alpha_{xy}(\tilde{P}_{u_{N}})\ge\alpha_{xy}(\tilde{P}_{\lambda})\qquad\mathrm{\text{and}\qquad var}\bigl(f,\tilde{P}_{u_{N}}\bigr)\le\mathrm{var}\bigl(f,\tilde{P}_{\lambda}\bigr),
\]
and the following functions from $\{1,\ldots,N\}$ to $\mathbb{R}_{+}$
satisfy: $k\mapsto\alpha_{xy}(\tilde{P}_{u_{k}})$ is non-decreasing
and $k\mapsto\mathrm{var}\bigl(g,\tilde{P}_{u_{k}}\bigr)$ is non-increasing. 

Furthermore, if $\pi$ is not concentrated on points, ${\rm Gap}_{R}(\tilde{P}_{u_{N}})\geq{\rm Gap}(\tilde{P}_{\lambda})$
and $k\mapsto{\rm Gap}_{R}(\tilde{P}_{u_{k}})$ is non-decreasing.
\end{cor} 
\begin{proof} 
Follows from Theorem \ref{thm:schur-pseudo} by observing that $u_{N}\prec\lambda$
and that $u_{k}\prec u_{k-1}$.
\end{proof} 

\begin{rem} 
Note that the convex order $(u_{k},Z)\le_{cx}(u_{k-1},Z)$ can be
obtained directly by applying \citep[Corollary 1.5.24]{mullercomparison},
which is related to the existence of reverse martingales for U-statistics,
but our result is slightly stronger.  
\end{rem}

The monotonicity result in Corollary \ref{cor:optimal-weights} provides
us with some justification for averaging. Because we do not 
quantify the benefit of increased averaging, we cannot provide 
guidelines what the optimal number of samples. Instead, we point an
interested reader to the recent related work of Sherlock, Thiery, Roberts and Rosenthal
\citep{sherlock-thiery-roberts-rosenthal} and Doucet, Pitt,
Deligiannidis and Kohn \citep{doucet-pitt-deligiannidis-kohn}
who give conditions for optimal acceptance rates in some contexts.
We note that when parallel
architectures are available, averaging may be cheap or free,
but our result also form the basis for the justification
of validity of adaptive MCMC algorithms which seek for an optimal
number of samples in the spirit of those proposed in a different context
\citep{andrieu-robert,gilks-roberts-sahu}. For example, it is possible
to consider algorithms which increase or decrease the number of samples
according to some rule, aiming to reach a pre-defined average acceptance
rate. 


\subsection{Stratification}
\label{sub:Stratification} 

In the context of Monte Carlo methods, stratification is a technique
which aims to reduce variance of estimators of expectations. It turns
out that stratification can also imply improved performance in terms
of convex order. We refer the reader to very recent and important
progress in this area \citep{goldstein2011stochastic,goldstein2012stochastic},
but we start here with a more specific and immediately applicable
result.

Approximate Bayesian computation (ABC) \citep{beaumont2002approximate,tavare1997inferring}
are now popular methods which are applicable in Bayesian inference
involving intractable (or expensive to evaluate) likelihood function,
but where simulation from the model is easy. Consider some data $y^{*}\in\mathsf{Y}$
and assume that it arises from a family of probability distributions
with densities $\bigl\{\ell\bigl(\cdot\mid x\bigr),x\in\mathsf{X}\bigr\}$,
with respect to some appropriate reference measure $\lambda$. Instead
of the exact likelihood $\ell\bigl(y^{*}\mid x\bigr)$, an approximate
likelihood function is constructed. Assume $s:\mathsf{Y}^{2}\rightarrow\mathbb{R}_{+}$
is a function whose r\^ole is to measure dissimilarity between datasets,
and consider for some $\epsilon>0$ the modified ABC likelihood 
\[
\ell_{{\rm ABC}}\bigl(y^{*}\mid x\bigr):=\int\ell\bigl(y\mid x\bigr)\mathbb{I}\bigl\{ s\bigl(y,y^{*}\bigr)\leq\epsilon\bigr\}\lambda\bigl({\rm d}y\bigr).
\]
This alternative likelihood function is in general intractable, but
naturally lends itself to pseudo-marginal computations \citep{andrieu-roberts}.

Indeed, for any $x\in\mathsf{X}$ assume $Y_{1},Y_{2},\ldots,Y_{N}\sim\ell\bigl(y\mid x\bigr)\lambda\bigl({\rm d}y\bigr)$
are independent samples. Then, one can construct a non-negative and
unbiased estimator $T_{x}$ of $\ell_{{\rm ABC}}\bigl(y^{*}\mid x\bigr)$
as follows 
\begin{equation}
T_{x}:=\frac{1}{N}\sum_{i=1}^{N}\mathbb{I}\bigl\{ s\bigl(Y_{i},y^{*}\bigr)\leq\epsilon\bigr\}.
\label{eq:abc-estimator}
\end{equation}
This leads to the unit expectation estimator $W_{x}=T_{x}/\ell_{{\rm ABC}}(y^{*}\mid x)$.
In practice, simulation of the random variables $Y$ on a computer
often involves using $d$ (pseudo-)random numbers uniformly distributed
on the unit interval $[0,1]$, which are then mapped to form one $Y_{i}.$
That is, there is a mapping from the unit cube $[0,1]^{d}$ to $\mathsf{Y}$,
and with an inconsequential abuse of notation, if $U_{i}\sim\mathcal{U}\bigl(\bigl[0,1\bigr]^{d}\bigr)$
then $Y\bigl(U_{i}\bigr)\sim\ell\bigl(y\mid x\bigr)\lambda\bigl({\rm d}y\bigr)$.
\begin{example}
An extremely simple illustration of this is the situation where $d=1$
and an inverse cdf method is used, that is $Y\bigl(U\bigr)=F^{-1}\bigl(U\bigr)$
where $F$ is the cumulative distribution function (cdf) of $Y$.
For example, in the case of the $g$-and-$k$ distribution the inverse
cdf is given by \citep{fearnhead-prangle}
\[
F^{-1}\bigl(u;A,B,c,g,k\bigr)=A+B\left(1+c\frac{1-\exp\bigl(-gz(u)\bigr)}{1+\exp\bigl(-gz(u)\bigr)}\right)\bigl(1+z(u)^{2}\bigr)^{k}z(u),
\]
where $z(u)$ is the standard normal quantile, and $A,B,c,g,k$ are
parameters.
\end{example}
In this context, an easily implementable method to improve performance
of the corresponding pseudo marginal algorithm is as follows. Let
$\mathcal{A}:=\bigl\{ A_{1},\ldots,A_{N}\bigr\}$ be a partition of
the unit cube $[0,1]^{d}$ such that $\mathbb{P}\bigl(U_{1}\in A_{i}\bigr)=1/N$,
and such that it is possible to sample uniformly from each $A_{i}$.
Perhaps the simplest example of this is when $\mathcal{A}$ corresponds
to the dyadic sub-cubes of $[0,1]^{d}$. Let $V_{i}\sim\mathcal{U}\bigl(A_{i}\bigr)$
for $i=1,\ldots,N$ be independent. We may now replace the estimator
in (\ref{eq:abc-estimator}) with 
\begin{equation}
T_{x}^{{\rm strat}}:=\frac{1}{N}\sum_{i=1}^{N}\mathbb{I}\bigl\{ s\bigl(Y(V_{i}),y^{*}\bigr)\leq\epsilon\bigr\}.
\label{eq:stratified-estimator}
\end{equation}
It is straightforward to check that this is a non-negative unbiased
estimator of $\ell_{{\rm ABC}}(y^{*}\mid x)$ which means that $W_{x}^{\mathrm{strat}}:=T_{x}^{\text{\ensuremath{\mathrm{strat}}}}/\ell_{{\rm ABC}}(y^{*}\mid x)$
has unit expectation as required. With $\tilde{P}$ the pseudo-marginal
approximation corresponding to using $\{W_{x}\}_{x\in\mathsf{X}}$
and $\tilde{P}^{{\rm strat}}$ the approximation corresponding to
$\{W_{x}^{\mathrm{strat}}\}_{x\in\mathsf{X}},$ we have the following
result.
\begin{thm}
For any $x\in\mathsf{X}$ we have $W_{x}^{\mathrm{strat}}\le_{cx}W_{x}$
and therefore for any $x,y\in\mathsf{X}$ and $f\in L^{2}\bigl(\mathsf{X},\pi\bigr)$,
\[
\alpha_{xy}\bigl(\tilde{P}^{\text{\ensuremath{\mathrm{strat}}}}\bigr)\ge\alpha_{xy}\bigl(\tilde{P}\bigr)\qquad\text{and\qquad}\mathrm{var}\bigl(f,\tilde{P}^{\mathrm{strat}\text{}}\bigr)\le\mathrm{var}\bigl(f,\tilde{P}\bigr).
\]
Furthermore, if $\pi$ is not concentrated on points, ${\rm Gap}_{R}\bigl(\tilde{P}^{\mathrm{strat}}\bigr)\ge{\rm Gap}_{R}\bigl(\tilde{P}\bigr)$.\end{thm}
\begin{proof}
Notice that $\mathbb{I}\bigl\{ s\bigl(Y\bigl(U_{i}\bigr),y^{*}\bigr)\leq\epsilon\bigr\}$
is a Bernoulli random variable of parameter $\bar{p}:=\mathbb{P}\bigl(s\bigl(Y(U_{i}),y^{*}\bigr)\leq\epsilon\bigr)=\mathbb{P}\bigl(U_{i}\in H\bigr)$
where $H:=\bigl\{ u\in\bigl[0,1\bigr]^{d}\,:\, s\bigl(Y(u),y^{*}\bigr)\leq\epsilon\bigr\}$.
Similarly, let $q_{i}:=\mathbb{P}\bigl(s\bigl(Y(V_{i}),y^{*}\bigr)\leq\epsilon\bigr)=\mathbb{P}\bigl(V_{i}\in H\bigr)=\mathbb{P}(U_{i}\in H\cap A_{i})/\mathbb{P}(U_{i}\in A_{i})$.
Note that $\bar{p}=\sum_{i=1}^{N}q_{i}/N$. Denoting $q=(q_{1},\ldots,q_{N})$
and $p=(\bar{p},\ldots,\bar{p})$, we have the majorisation $p\prec q$
(see Definition \ref{def:majorizationschurconcavity} in Section \ref{sec:Averaging-in-the}).
A well known result \citep{hoeffding1956} and \citep{karlin1963}
tells us that the sum of the corresponding Bernoulli random variables
are convex ordered, that is for independent $W_{1},\ldots,W_{N}\sim\mathcal{U}\bigl(0,1\bigr)$,
\[
\frac{1}{N}\sum_{i=1}^{N}\mathbb{I}\bigl\{ W_{i}\leq q_{i}\bigr\}\lecx \frac{1}{N}\sum_{i=1}^{N}\mathbb{I}\bigl\{ W_{i}\leq p_{i}\bigr\}.
\]
The random variable on the left coincides in distribution with $T_{x}^{\mathrm{strat}}$
and the random variable on the right coincides with $T_{x}.$ Consequently,
$W_{x}^{\mathrm{strat}}=T_{x}^{\mathrm{strat}}/\ell_{\mathrm{ABC}}(y^{*}\mid x)\le_{cx}T_{x}/\ell_{\mathrm{ABC}}(y^{*}\mid x)=W_{x}.$
\end{proof}
Naturally some stratification schemes are going to be better than
others, and the majorisation characterisation provides us, in principle,
with a criterion for comparisons.
\begin{rem}
We note that in some contexts, stratification may also open the possibility
for additional computational savings. First, using ``early rejection''
as suggested in \citep{solonen-ollinaho-laine-haario-tamminen-jarvinen}
the values of the summands in (\ref{eq:stratified-estimator}) can
be computed progressively until it is possible to decide whether the
sample is accepted or rejected. Second, in certain scenarios it may
be possible to deduce the values of some indicators in (\ref{eq:stratified-estimator}),
before computing them, from realisations of others. For example, assume
$d=1$, $A_{i}=[(i-\text{1})/N,i/N]$, $s(y,y^{*})=|y-y^{*}|,$ and
suppose we use the inverse cdf method. Then due to the monotonicity
of the inverse cdf, if $\mathbb{I}\bigl\{ s\bigl(Y(V_{k}),y^{*}\bigr)\leq\epsilon\bigr\}=1$
and $\mathbb{I}\bigl\{ s\bigl(Y(V_{k+1}),y^{*}\bigr)\leq\epsilon\bigr\}=0$,
then we know that necessarily $\mathbb{I}\bigl\{ s\bigl(Y(V_{i}),y^{*}\bigr)\leq\epsilon\bigr\}=0$
for $i=k+2,\ldots,N$.
\end{rem}

\begin{rem}
It has also been suggested in the literature to replace the indicator
function in the ABC likelihood with a more general ``kernel'' $K:\mathbb{R}_{+}\rightarrow[0,1]$
effectively leading, with $\psi=K\circ s$, to 
\[
\ell_{{\rm ABC}}^{\psi}\bigl(y^{*}\mid x\bigr):=\int\ell\bigl(y\mid x\bigr)\psi\bigl(y,y^{*}\bigr)\lambda\bigl({\rm d}y\bigr)\quad.
\]
In such a situation it is still possible to use stratification, but
now inter-related conditions on the stratification scheme and the
mapping $u\mapsto\psi\bigl(Y(u),y^{*}\bigr)$ are needed. For example,
in the scenario $d=1$ and with a monotone partition, the mapping
should be monotone, otherwise the sought convex order may not hold
\citep{goldstein2011stochastic}.
\end{rem}


\subsection{Extremal properties}
\label{sec:Extremal-properties} 

In this section we investigate upper and lower bounds on the performance
of pseudo-marginal algorithms by establishing a, perhaps surprising,
link to the actuarial science literature in terms of extremal moments
and stop-loss functions \citep{de1982analytical,janssens2008use,hurlimann2008extremal}.
More specifically we consider unit expectation distributions $Q_{*}$
and $Q^{*}$ which are minimal and maximal in the convex orders, subject
to some constraints. We focus particularly on two types of constraints:
a support constraint or a variance constraint. Other constraints such
as a kurtosis constraint or a modality constraint are possible, and
an interested reader can consult \citep{hurlimann2008extremal}.

The link to the actuarial science literature comes from the fact that
convex order of distributions of random variables $W$ and $V$ with
$\mathbb{E}W=\mathbb{E}V$ is determined by the order of the related
stop-loss functions (or integrated survival functions) $\mathbb{E}\big[(W-t)_{+}\big]$
and $\mathbb{E}\big[(V-t)_{+}\big]$; see Lemma \ref{lem:cx-char}.
The stop-loss links directly with the expected acceptance probability
of the algorithm in Example \ref{ex:toyP_ring} through the identity
\begin{align*}
\min\left\{ 1,r(x,y)\varpi\right\} &=r(x,y)\varpi-\max\left\{
0,r(x,y)\varpi-1\right\} \\
&=r(x,y)\varpi-r(x,y)\big(\varpi-r^{-1}(x,y)\big)_{+}.
\end{align*}
In the case of the pseudo-marginal algorithm, the above identity with
$\varpi=u/w$ provides a connection; see the proof of Theorem \ref{thm:propertiesmoon_Ps}.

Let $\mathscr{P}$ be some subset of probability distributions on
$(\mathbb{R},\mathcal{B}(\mathbb{R}))$. Well researched questions
about stop-losses involve determining extremal elements $Q\in\mathscr{P}$
maximising or minimising $\mathbb{E}_{Q}\big[(W-t)_{+}\big]$ for
some or all $t\in\mathbb{R}$. We review some of these results particularly
relevant to the present set-up and apply them to our problem. 
\begin{thm}
\label{thm:minimum-cx}Let $\mu\in\mathbb{R}$ and let $\mathscr{P}(\mu)$
stand for the probability distributions $Q$ on $\mathbb{R}$ such
that the random variable $W\sim Q$ has expectation $\mathbb{E}_{Q}\bigl[W\bigr]=\mu$.
Then, for any $t\in\mathbb{R},$ 

\begin{align*}
\delta_{\mu}(\mathrm{d}w) & =\mathop\mathrm{arg\, min}_{Q\in\mathscr{P(\mu)}}\mathbb{E}_{Q}\big[(W-t)_{+}\big],
\end{align*}
with minimum value $(\mu-t)_{+}$.
\end{thm}

\begin{thm}
\label{thm:extremum_interval_ab}Let $a,b,\mu\in\mathbb{R}$ with
$a\leq\mu\leq b$ and let $\mathscr{P}(\mu,[a,b])\subset\mathscr{P}(\mu)$
be the set of probability distributions $Q$ on $[a,b],$ that is,
satisfying $\mathbb{E}_{Q}[W]=\mu$ and $Q\big([a,b]\big)=1.$ Then
for any $t\in\mathbb{R}$, 
\[
\frac{b-\mu}{b-a}\delta_{a}(\mathrm{d}w)+\frac{\mu-a}{b-a}\delta_{b}(\mathrm{d}w)=\mathop\mathrm{arg\, max}_{Q\in\mathscr{P}(\mu,[a,b])}\mathbb{E}_{Q}\big[(W-t)_{+}\big]
\]
 with maximum value\textup{ $\frac{b-\mu}{b-a}(a-t)_{+}+\frac{\mu-a}{b-a}(b-t)_{+}.$}
\end{thm}
\noindent The proofs of Theorems \ref{thm:minimum-cx} and \ref{thm:extremum_interval_ab}
can be found in \citep{de1982analytical,janssens2008use,hurlimann2008extremal}.

We state two direct consequences of these results.
\begin{thm}
Let $a_{x},b_{x}\in\mathbb{R}_{+}$ be such that $a_{x}\leq1\leq b_{x}$
for all $x\in\mathsf{X}$. Consider the class of pseudo-marginal algorithms
$\tilde{P}$ such that for any $x\in\mathsf{X}$ the weight distribution
$Q_{x}$ is concentrated on $[a_{x},b_{x}],$ that is, $Q_{x}\in\mathscr{P}(1,[a_{x},b_{x}])$.
Then, for any $f\in L^{2}\bigl(\mathsf{X},\pi\bigr)$,
\[
{\rm var}\left(f,P\right)\leq{\rm var}\bigl(f,\tilde{P}\bigr)\leq{\rm var}\bigl(f,\tilde{P}_{{\rm max}}\bigr),
\]
where $\tilde{P}_{{\rm max}}$ is the pseudo-marginal algorithm with
noise distributions
\[
Q_{x}^{\max}(\mathrm{d}w)=\frac{1-a_{x}}{b_{x}-a_{x}}\delta_{a_{x}}(\mathrm{d}w)+\frac{b_{x}-1}{b_{x}-a_{x}}\delta_{b_{x}}(\mathrm{d}w).
\]
Furthermore,
\[
{\rm var}\bigl(f,\tilde{P}_{{\rm max}}\bigr)\leq\sup_{x\in\mathsf{X}}b_{x}\mathrm{var}\left(f,P\right)+(\sup_{x\in\mathsf{X}}b_{x}-1)\mathrm{var}_{\pi}(f).
\]
\end{thm}
\begin{proof}
The first claim is direct from Theorems \ref{thm:minimum-cx} and
\ref{thm:extremum_interval_ab}, Lemma \ref{lem:cx-char} and Lemma
\ref{thm:orderingpseudomarginals}. The last claim follows from \citep[Corollary 11]{andrieu-vihola-2012}.
\end{proof}
We next state for completeness a similar result for algorithms $\mathring{P}$
as discussed in Section \ref{sec:Exact-approximations-ratios}. In
particular, it is direct to check that the diatomic distributions
in Example \ref{ex:toyP_ring} of the form, with $a_{xy}=a_{yx}\ge1$,
\[
Q_{xy}(\mathrm{d}\varpi)=\frac{a{}_{xy}}{1+a_{xy}}\delta_{a_{xy}^{-1}}(\mathrm{d}\varpi)+\frac{1}{1+a_{xy}}\delta_{a_{xy}}(\mathrm{d}\varpi),
\]
are maximal among those with support on $[a_{xy}^{-1},a_{xy}]$. We
quote the result without a proof, as it is a direct consequence of
the convex order property and Peskun's result.
\begin{thm}
Let $a_{xy}\in[1,\infty)$ be constants such that $a_{xy}=a_{yx}$
for all $x,y\in\mathsf{X}$. Consider any algorithm $\mathring{P}$
as in Section \ref{sec:Exact-approximations-ratios} such that $Q_{xy}\in\mathscr{P}(\mu=1,[a_{xy}^{-1},a_{xy}])$
for all $x,y\in\mathsf{X}$. Then, for any $x,y\in\mathsf{X}^{2}$
\begin{align*}
\frac{a_{xy}}{1+a_{xy}}\min\left\{ 1,r(x,y)a_{xy}^{-1}\right\}
&+\frac{1}{1+a_{xy}}\min\left\{ 1,r(x,y)a_{xy}\right\} \\
&\leq\int Q_{xy}(\mathrm{d}\varpi)\min\left\{ 1,r(x,y)\varpi\right\} \\
&\leq\min\left\{ 1,r(x,y)\right\},
\end{align*}
and for any $f\in L^{2}(\mathsf{X},\pi)$,
\begin{align*}
\mathrm{var}\left(f,P\right)&\leq{\rm
var}\bigl(f,\mathring{P}\bigr)\\
&\leq{\rm var}\bigl(f,\mathring{P}_{{\rm
max}}\bigr)\\
&\leq\sup_{x,y\in\mathsf{X}^{2}}a_{xy}\mathrm{var}\left(f,P\right)+(\sup_{x,y\in\mathsf{X}^{2}}a_{xy}-1)\mathrm{var}_{\pi}(f).
\end{align*}

\end{thm}
We now turn back to pseudo-marginal algorithms. Not surprisingly,
it is impossible to find a maximal distribution on $\mathscr{P}(\mu,[a,b])$
whenever either $a=-\infty$ or $b=\infty$. However, as we shall
see, with an additional constraint on the variance $\sigma^{2}<\infty$
of the distributions, it is possible to find a supremal distribution
even when $b=\infty$. More specifically, the stop-loss function can
be maximised, but the corresponding class of distributions is not
closed and maximising distribution will not have a finite variance.
We first state the following results which can be found in \citep{de1982analytical,janssens2008use,hurlimann2008extremal}.
\begin{thm}
\label{thm:maximal-cx-bounded-variance}Let $a,b,\mu\in\mathbb{R}$
such that $a\leq\mu\leq b$ and let $\sigma^{2}\in\big[0,(\mu-a)(b-\mu)\big]$.
Define $\mathscr{P}(\mu,\sigma^{2},[a,b])\subset\mathscr{P}(\mu,[a,b])$
be the set of probability distributions $Q$ such that $\mathbb{E}_{Q}[W]=\mu$,
$\mathrm{var}_{Q}(W)=\sigma^{2}$ and $Q\big([a,b]\big)=1.$ Denote
$\sigma_{\mu}^{2}(t):=\sigma^{2}+(\mu-t)^{2}$ and $c:=(a+b)/2$.
Then, the maximisation problem
\[
Q^{*}:=\mathop\mathrm{arg\, max}_{Q\in\mathscr{P}(\mu,\sigma^{2},[a,b])}\mathbb{E}_{Q}\big[(W-t)_{+}\big],
\]
has the following solutions for different values of $t$, where $Q^{*}$
is a diatomic distribution with the given atoms,

\begin{center}
\begin{tabular}{ccc}
\toprule 
$\mathbb{E}_{Q*}\big[(W-t)_{+}\big]$  & Atoms of $Q^{*}$ & Range of $t$\tabularnewline
\midrule
$\frac{1}{2}\bigl(\mu-t+\sigma_{\mu}(t)\bigr)$ & $t-\sigma_{\mu}(t),\, t+\sigma_{\mu}(t)$ & $t\leq c,\,\sigma_{\mu}(t)\leq t-a$\tabularnewline
$\bigl(\mu-a\bigr)\frac{(\mu-t)(\mu-a)+\sigma^{2}}{(\mu-a)^{2}+\sigma^{2}}$ & $a,\,\mu+\frac{\sigma^{2}}{\mu-a}$ & $t\leq c,\,\sigma_{\mu}(t)\geq t-a$\tabularnewline
$\frac{1}{2}\bigl(\mu-t+\sigma_{\mu}(t)\bigr)$ & $t-\sigma_{\mu}(t),\, t+\sigma_{\mu}(t)$ & $t\geq c,\,\sigma_{\mu}(t)\leq b-t$\tabularnewline
$\frac{(b-t)\sigma^{2}}{(\mu-b)^{2}+\sigma^{2}}$ & $\mu-\frac{\sigma^{2}}{b-\mu},\, b$ & $t\geq c,\,\sigma_{\mu}(t)\geq b-t$\tabularnewline
\bottomrule
\end{tabular}
\par\end{center}

\end{thm}

\noindent The following result is a restatement of \citep[Theorem 1.5.10, b)]{mullercomparison},
and it gives us a way to extend Theorem \ref{thm:maximal-cx-bounded-variance}
to unbounded supports.
\begin{thm}
\label{thm:muller-stop-loss}Suppose $\phi:\mathbb{R}\to\mathbb{R}$
is non-increasing and convex, and satisfies $\lim_{t\to\infty}\phi(t)=0$
and $\lim_{t\to-\infty}\phi(t)=\mu\in\mathbb{R}.$ Then, there exists
a random variable $X$ with $\mathbb{E}[X]=\mu$ such that $\phi(t)=\mathbb{E}\big[(X-t)_{+}\big],$
and the cdf of $X$ can be written as $F_{X}(t)=1+\phi'(t),$ where
$\phi'$ stands for the right derivative of $\phi$.\end{thm}
\begin{example}
In our case, we are interested in distributions on the positive real
line and with unit mean, for which we notice that
\begin{align*}
\phi(t)&:=\lim_{b\to\infty}\sup_{\mathscr{P}(1,\sigma^{2},[0,b])}\mathbb{E}\bigl((W-t)_{+})\\
&=\begin{cases}
\frac{(1-t)+\sigma^{2}}{1+\sigma^{2}}, & 0\leq t\leq\frac{\sigma^{2}+1}{2},\\
\frac{1}{2}\big(\sqrt{\sigma^{2}+(1-t)^{2}}+1-t\big),\; & t\geq\frac{\sigma^{2}+1}{2}.
\end{cases}
\end{align*}
It is straightforward to check that $\phi$ satisfies the conditions
in Theorem \ref{thm:muller-stop-loss}, so the ``supremal'' distribution
$Q_{\sigma^{2}}^{*}$ on $\mathscr{P}\big(1,\sigma^{2},[0,\infty)\big)$
has the following cumulative distribution function 
\[
Q_{\sigma^{2}}^{*}\big([0,t]\big)=\begin{cases}
\frac{\sigma^{2}}{1+\sigma^{2}}, & 0\leq t<\frac{\sigma^{2}+1}{2},\\
\frac{1}{2}+\frac{1}{2}\frac{t-1}{\sqrt{\sigma^{2}+(1-t)^{2}}},\; & t\geq\frac{\sigma^{2}+1}{2}.
\end{cases}
\]
Note that $Q_{\sigma^{2}}^{*}$ does not belong to $\mathscr{P}\big(1,\sigma^{2},[0,\infty)\big)$
because it has infinite variance, but the expectation is one, which
can be verified by noting that $\phi(0)=1.$

\end{example}



\section{Discussion and perspectives}
\label{sec:Discussion-and-perspectives} 

In this paper we have shown that the convex (partial) order of distributions
is a natural and useful tool in order to compare various performance
measures of competing implementations of exact approximations of the
Metropolis-Hastings algorithms. As examples of applications of our
theory, we have shown that it is possible to identify extremal behaviours
of such algorithms under various distributional constraints. More
importantly from a practical point of view, we have shown that averaging
of independent estimators improves performance monotonically. Even
though averaging may not always be useful in the context of sequential
implementations, it turns out to be particularly relevant when parallel
architectures are available \citep[e.g.][]{drovandi2014}. Prompted
by our theory and other recently established results, we have also
proposed to use stratification in ABC MCMC and beyond, which has the
advantage to provide better performance at no additional computational
cost. 

There are many other results from the stochastic ordering literature
relevant to the present context we have not yet investigated. For
example, introducing negative dependence when averaging two weights
could further improve performance. This is a direct application of
the result on positive and negative quadrant dependence of \citep[Lemma 2]{dhaene1996dependency}.
Other dependence orders could be exploited, such as the supermodular
order \citep{mullercomparison} which can be used to characterise
the (positive) dependence order  of the components of random vectors
of arbitrary length. In this scenario the supermodular order 
\[
\big(W_{1}^{(1)},W_{2}^{(1)},\ldots,W_{N}^{(1)}\big)
\leq_\mathrm{sm}
\big(W_{1}^{(2)},W_{2}^{(2)},\ldots,W_{N}^{(2)}\big),
\]
implies the convex order 
$
\sum_{i=1}^{N}W_{i}^{(1)}\lecx \sum_{i=1}^{N}W_{i}^{(2)},
$
see, for example, the results in \citep[Section 9.A]{shaked-shanthikumar}.

We would like to point out here another promising and useful avenue
of research related to the discussion of \citep{andrieu-doucet-holenstein}
by Lee and Holmes to which our current theory does not seem to apply
directly. First, we notice that pseudo-marginal algorithms can be
extended to the situation where we can define a joint distribution
$Q_{xy}({\rm d}w\times{\rm d}u)$ with marginals $Q_{x}({\rm d}w)$
and $Q_{y}({\rm d}u)$ and which satisfies the following symmetry
condition for any $x,y\in\mathsf{X}$ and $A,B\in\mathcal{B}(\mathbb{R}_{+})$,
$Q_{xy}\bigl(A\times B\bigr)=Q_{yx}(B\times A)$.
The proposal distribution used in this algorithm is the corresponding
conditional distribution $Q_{xy}\bigl({\rm d}u|w\bigr)$, which now
depends on $x$, $y$ and $w$, and the acceptance ratio remains as
in (\ref{eq:kernelpseudomarginal}). Standard pseudo-marginal algorithms
correspond to the choice $Q_{xy}({\rm d}w\times{\rm d}u)=Q_{x}({\rm d}w)Q_{y}({\rm d}u)$.
This formalism allows one to disentangle the dependence structure
from the variability of the marginal distributions. One can easily
establish that this results in a MH kernel with $\tilde{\pi}$ as
invariant distribution. 

Intuitively, inducing positive dependence should reduce the variability
of the acceptance probability and therefore lead to better performance---this
is the motivation behind the work of Lee and Holmes. We note however that
this is likely to reduce ``mixing'' on the noise component $w$. An order which seems suitable to rank
such algorithms is the concordance order, also known as the correlation
order; see \citep{dhaene1996dependency}, which coincides with the
upper orthant, concordance and supermodular order in the bivariate
scenario \citep{mullercomparison}.
Using, for simplicity, our earlier notation for the present scenario,
one can show that if for some $x,y\in\mathsf{X}$
$(W_{xy}^{(2)},U_{xy}^{(2)})\leq_\mathrm{c}(W_{xy}^{(1)},U_{xy}^{(1)})$ then
$\alpha_{xy}\bigl(\tilde{P}^{(2)}\bigr)\leq\alpha_{xy}\bigl(\tilde{P}^{(1)}\bigr)$
and therefore
$\mathcal{E}_{\tilde{P}^{(1)}}\bigl(g\bigr)\geq\mathcal{E}_{\tilde{P}^{(2)}}\bigl(g\bigr)$
for any $g \in L^{2}\bigl(\mathsf{X},\pi\bigr)$. However we do not
know whether or not this implies ${\rm
var}\bigl(f,\tilde{P}^{(1)}\bigr)\leq{\rm
var}\bigl(f,\tilde{P}^{(2)}\bigr)$ for $f \in
L^{2}\bigl(\mathsf{X},\pi\bigr)$.

\section*{Acknowledgements} 

The authors would like to thank Lasse Leskel\"a for illuminating conversations
about convex orders and peacock processes. 

Christophe Andrieu was supported by the EPSRC programme grant ``Intractable
Likelihood: New Challenges from Modern Applications (ILike)'' and
the EPSRC grant ``Bayesian Inference for Big Data with Stochastic
Gradient Markov Chain Monte Carlo''. Matti Vihola was supported by the 
Academy of Finland fellowships 
250575 and 274740.


\appendix

\section{Proof of Lemma \ref{lem:variationalrepinverseoperator}}
\label{app:caracciolo-pelissetto-sokal} 

It is easy to see that also $A^{-1}$ is self-adjoint, and for any
$g\in\mathcal{H}$, 
\begin{align*}
    0 &\le \big\langle A^{-1} f - g, A(A^{-1} f - g)\big\rangle 
       = \langle f, A^{-1} f \rangle - 2 \langle f, g\rangle +
      \langle g, A g \rangle.
\end{align*}
This implies \eqref{eq:bellman} with inequality `$\ge$'. We conclude
by taking $g = A^{-1} f$.
\qed 


\section{Perturbed Metropolis-Hastings algorithms} 
\label{sec:perturbed-mh} 

Assume $\pi$, $q$ and $r$ are as defined
in Section \ref{sec:Introduction}, and assume $\mathring{\pi}_{x_{1}}(\mathrm{d}x_{2})$
are probability distributions on $(\mathsf{X_{2}},\mathcal{X}_{2})$
for all $x_{1}\in\mathsf{X},$ and denote with a slight abuse of notation
the distribution $\mathring{\pi}({\rm d}x_{1}\times\mathrm{d}x_{2})=\pi(\mathrm{d}x_{1})\mathring{\pi}_{x_{1}}({\rm d}x_{2})$
on $(\mathring{\mathsf{X}},\mathring{\mathcal{X}})=(\mathsf{X\times}\mathsf{X_{2}},\mathcal{X}\times\mathcal{X}_{2}).$
Assume that we are interested in approximating a Metropolis-Hastings
algorithm targeting $\mathring{\pi}({\rm d}x_{1}\times\mathrm{d}x_{2})$
using a family of proposal distributions of the following form, with
$x:=\bigl(x_{1},x_{2}\bigr)\in\mathsf{\mathring{\mathsf{X}}}$, 
\[
\mathring{q}(x,\mathrm{d}y_{1}\times\mathrm{d}y_{2}):=q(x_{1},{\rm d}y_{1})\mathring{\pi}_{y_{1}}({\rm d}y_{2}).
\]
 This covers, for example, the algorithm used in \citep{karagiannis-andrieu-2013}.
Note that this algorithm is also of the type discussed in Appendix
\ref{sec:augmented}.
\begin{lem}
\label{lem:circlePreversible}Let $\bigl\{ Q_{xy_{1}}({\rm d}\varpi\times{\rm d}y_{2})\bigr\}_{(x,y_{1})\in\mathsf{\mathring{\mathsf{X}}\times\mathsf{X}}}$
be a family of probability distributions on $((0,\infty)\times\mathsf{X_{2}},\mathcal{B}\bigl((0,\infty)\bigr)\times\mathcal{X}_{2})$.
Consider the Markov transition kernel \textup{$\mathring{P}$} on
$\mathsf{(\mathring{\mathsf{X}},\mathcal{\mathring{\mathcal{X}}})}$
defined through 
\[
\mathring{P}(x;\mathrm{d}y_{1}\times\mathrm{d}y_{2}):=q(x_{1},\mathrm{d}y_{1})\int Q_{xy_{1}}(\mathrm{d}\varpi\times\mathrm{d}y_{2})\min\left\{ 1,r(x_{1},y_{1})\varpi\right\} +\delta_{x}(\mathrm{d}y)\mathring{\rho}(x),
\]
where the probability of rejection $\mathring{\rho}(x)\in[0,1]$ is
such that $\mathring{P}\bigl(x,\,\cdot\,\bigr)$ defines a probability
distribution on $\mathsf{(\mathring{\mathsf{X}}},\mathcal{\mathring{\mathcal{X}}})$
for all $x\in\mathring{\mathsf{X}}$. Assume further that for any
$x_{1},y_{1}\in\mathsf{X}$ and any $A,B\in\mathcal{X}_{2}$ and $C\in\mathcal{B}\bigl((0,\infty)\bigr)$,
\begin{align*}
\int\mathring{\pi}_{x_{1}}(\mathrm{d}x_{2})Q_{xy_{1}}&\bigl(\mathrm{d}\varpi\times\mathrm{d}y_{2}\bigr)  \varpi\mathbb{I\Big\{}x_{2}\in A,\, y_{2}\in B,\,\frac{1}{\varpi}\in C\Big\}\\
 & =\int\mathring{\pi}_{y_{1}}(\mathrm{d}y_{2})Q_{yx_{1}}\bigl(\mathrm{d}\varpi\times\mathrm{d}x_{2}\bigr)\mathbb{I}\{x_{2}\in A,\, y_{2}\in B,\,\varpi\in C\}.
\end{align*}
Then $\mathring{P}$ is reversible with respect to $\mathring{\pi}$\textup{.}\end{lem}
\begin{proof}
Let $A,B\in\mathring{\mathcal{X}},$ then we may write, restricting
the integrals below to the set $\{r(x_{1},y_{1})>0\}$ \citep[see][]{tierney-note},
\begin{align*}
 \int\mathring{\pi}&(\mathrm{d}x)q(x_{1},\mathrm{d}y_{1})\int Q_{xy_{1}}(\mathrm{d}\varpi\times\mathrm{d}y_{2})\min\left\{ 1,r(x_{1},y_{1})\varpi\right\} \mathbb{I}\{x\in A,\: y\in B\}\\
 &
 =\int\pi(\mathrm{d}x_{1})q(x_{1},\mathrm{d}y_{1})r(x_{1},y_{1})\int\mathring{\pi}_{x_{1}}(\mathrm{d}x_{2})Q_{xy_{1}}\bigl(\mathrm{d}\varpi\times\mathrm{d}y_{2}\bigr)\varpi\\
 &\phantom{=}\times \min\left\{ \frac{r(y_{1},x_{1})}{\varpi},1\right\} 
 \mathbb{I}\{x_{2}\in A_{x_{1}},\: y_{2}\in B_{y_{1}}\}\\
 &
 =\int\pi(\mathrm{d}y_{1})q(y_{1},\mathrm{d}x_{1})\int\mathring{\pi}_{y_{1}}(\mathrm{d}y_{2})Q_{yx_{1}}\bigl(\mathrm{d}\varpi\times\mathrm{d}x_{2}\bigr)\min\left\{
 1,r(y_{1},x_{1})\varpi\right\} \\
 &\phantom{=}\times \mathbb{I}\{x_{2}\in A_{x_{1}},\: y_{2}\in B_{y_{1}}\},
\end{align*}
where $\, A_{x_{1}}:=\{x_{2}\in\mathsf{X}_{2}:(x_{1},x_{2})\in A\}$
and $B_{y_{1}}:=\{y_{2}\in\mathsf{X}_{2}:(y_{1},y_{2})\in B\}.$ We
can now conclude.
\end{proof}


\section{Spectral gaps of augmented kernels}
\label{sec:augmented} 

Recall that the left spectral gap of a $\mu$-reversible Markov kernel
$\Pi$ can be defined as 
\[
{\textstyle {\rm Gap}_{L}}(\Pi):=\inf_{\|f\|_{\mu}=1}\big\langle f,(I+\Pi_{\nu})f\big\rangle_{\mu},
\]
and the absolute spectral gap is defined as ${\rm Gap}(\Pi):=\min\{{\rm Gap}_{L}(\Pi),{\rm Gap}_{R}(\Pi)\}\in[0,1]$.
\begin{lem}
\label{lem:augmented-gaps}Assume $\Pi$ is a Markov kernel on a measurable
space $(\mathsf{E},\mathcal{F})$ and reversible with respect to a
probability measure $\mu$. Suppose $\Pi$ has the form 
\[
\Pi(x,{\mathrm{d}}y)=p(x,{\mathrm{d}}y)+\delta_{x}({\mathrm{d}}y)r(x),
\]
 where $p(x,{\mathrm{d}}y)\ge0$ is a sub-probability kernel and $r(x)\in[0,1]$
for all $x\in\mathsf{E}$. Assume $\nu(x,A)$ is a probability kernel
from $(\mathsf{E},\mathcal{F})$ to another measurable space $(\mathsf{S},\mathcal{S})$,
and define the Markov kernel 
\[
\Pi_{\nu}(x,w;{\mathrm{d}}y\times{\mathrm{d}}u):=p(x,{\mathrm{d}}y)\nu(y,{\mathrm{d}}u)+\delta_{x,w}({\mathrm{d}}y\times{\mathrm{d}}u)r(x).
\]
 Then, denoting $r^{*}:=\mathrm{\mu-ess\, sup}_{x}r(x)$ and $r_{*}:=\mathrm{\mu-ess\, inf}_{x}r(x),$
\begin{enumerate}[label=(\alph*)]
\item \label{item:aug-reversibility} $\Pi_{\nu}$ is reversible with respect
to $\mu_{\nu}({\mathrm{d}}x\times{\mathrm{d}}w)=\mu({\mathrm{d}}x)\nu(x,{\mathrm{d}}w)$, 
\item \label{item:aug-dirichlet} $\mathcal{E}_{\Pi}(f)=\mathcal{E}_{\Pi_{\nu}}(f)$
for all $f\in L^{2}(\mathsf{E,\mu)}$, with $f(x,w)=f(x)$,
\item \label{item:aug-right-gap} $\min\{{\textstyle {\rm Gap}_{R}}(\Pi),1-r^{*}\}\le{\textstyle {\rm Gap}_{R}}(\Pi_{\nu})\le{\textstyle {\rm Gap}_{R}}(\Pi)$, 
\item \label{item:aug-left-gap} $\min\{{\textstyle {\rm Gap}_{L}}(\Pi),1+r_{*}\}\le{\textstyle {\rm Gap}_{L}}(\Pi_{\nu})\le{\textstyle {\rm Gap}_{L}}(\Pi)$. 
\end{enumerate}
\end{lem}
\begin{proof}
The reversibility \ref{item:aug-reversibility} follows from \citep[cf.][]{tierney-note}
\[
\mu_{\nu}({\mathrm{d}}x\times{\mathrm{d}}w)p(x,{\mathrm{d}}y)\nu(y,{\mathrm{d}}u)=\mu_{\nu}({\mathrm{d}}y\times{\mathrm{d}}u)p(y,{\mathrm{d}}x)\nu(x,{\mathrm{d}}w).
\]
Consider then a function $f\in L^{2}(\mathsf{E}\times\mathsf{S},\mu_{\nu})$.
We write any such $f$ as $f=\bar{f}+f_{0}$ where $f_{0}(x,w)=f_{0}(x):=\int\nu(x,{\mathrm{d}}w)f(x,w)$
and $\bar{f}=f-f_{0}$. It is straightforward to check that $\Pi_{\nu}f_{0}(x,w)=\Pi f_{0}(x)$
and that $\Pi_{\nu}f=\Pi f_{0}+r\bar{f}$, implying $\Pi_{\nu}\bar{f}(x,w)=r\bar{f}$.
These allow us to write 
\begin{align}
\label{eq:aug-dirichlet}
\big\langle f,\Pi_{\nu}f\big\rangle_{\mu_{\nu}}&=\big\langle
f_{0},\Pi_{\nu}f_{0}\big\rangle_{\mu_{\nu}}+\big\langle\bar{f},\Pi_{\nu}\bar{f}\big\rangle_{\mu_{\nu}}+2\big\langle
f_{0},\Pi_{\nu}\bar{f}\big\rangle_{\mu_{\nu}}\\
&=\big\langle f_{0},\Pi
f_{0}\big\rangle_{\mu}+\langle\bar{f},r\bar{f}\rangle_{\mu_{\nu}}.\nonumber
\end{align}
For $f$ constant in the second variable, we have $f_{0}=f$ and $\bar{f}=0$
and $\|f_{0}\|_{\mu}=\|f\|_{\mu_{\nu}}$, implying \ref{item:aug-dirichlet}.

For the spectral gap bounds, assume $\|f\|_{\mu_{\nu}}=1$ and note
that $1=\|f\|_{\mu_{\nu}}^{2}=\|f_{0}\|_{\mu_{\nu}}^{2}+\|\bar{f}\|_{\mu_{\nu}}^{2}$.
This, with \eqref{eq:aug-dirichlet}, allows us to deduce 
\begin{align*}
{\textstyle {\rm
Gap}_{R}}(\Pi_{\nu})&=\inf_{\|f\|_{\mu_{\nu}}=1}\big\langle
f,(I-\Pi_{\nu})f\big\rangle_{\mu_{\nu}}\\
&=\inf_{\|f\|_{\mu_{\nu}}=1}1-\big\langle
f,\Pi_{\nu}f\big\rangle_{\mu_{\nu}}\\
&=\inf_{\|f\|_{\mu_{\nu}}=1}\big\langle f_{0},(I-\Pi)f_{0}\big\rangle_{\mu}+\langle\bar{f},(1-r)\bar{f}\rangle_{\mu_{\nu}},
\end{align*}
where the $\inf$ and $\sup$ are taken over functions $f\in L_{0}^{2}(\mathsf{E}\times\mathsf{S},\mu_{\nu})$.
The right hand side inequality in \ref{item:aug-right-gap} follows
by restricting to functions constant in the second variable. For the
first inequality, note that $\langle\bar{f},(1-r)\bar{f}\rangle_{\mu_{\nu}}\ge(1-r^{*})\|\bar{f}\|_{\mu_{\nu}}^{2}$
and $\big\langle f_{0},(I-\Pi)f_{0}\big\rangle_{\mu}\ge\mathrm{Gap}(\Pi)\|f_{0}\|_{\mu}^{2}$.
The claim follows because $\|f_{0}\|_{\mu_{\nu}}^{2}+\|\bar{f}\|_{\mu_{\nu}}^{2}=1$.

Similarly for the left gap \ref{item:aug-left-gap}, 
\begin{align*}
{\textstyle {\rm
Gap}_{L}}(\Pi_{\nu})&=\inf_{\|f\|_{\mu_{\nu}}=1}\big\langle
f,(I+\Pi_{\nu})f\big\rangle_{\mu_{\nu}}\\
& =\inf_{\|f\|_{\mu_{\nu}}=1}\big\langle f_{0},(I+\Pi)f_{0}\big\rangle_{\mu}+\langle\bar{f},(1+r)\bar{f}\rangle_{\mu_{\nu}},
\end{align*}
and $\langle\bar{f},(1+r)\bar{f}\rangle_{\mu_{\nu}}\ge(1+r_{*})\|\bar{f}\|_{\mu_{\nu}}^{2}$.\end{proof}
\begin{rem}
\label{rem:continuous-gaps}In the following special scenarios some
of the conclusions of Lemma \ref{lem:augmented-gaps} take a simpler
form.
\begin{enumerate}[label=(\alph*)]
\item \label{item:gap-continuous-pi}If $\mu$ is not concentrated on points,
that is, $\mu(\{x\})=0$ for all $x\in\mathsf{E}$, then $\mathrm{Gap}_{R}(\Pi)\le1-r^{*}$
and therefore ${\textstyle {\rm Gap}_{R}}(\Pi_{\nu})={\textstyle {\rm Gap}_{R}}(\Pi).$
\item \label{item:augmented-positivity}If $\Pi$ is positive, that is,
$\big\langle g,\Pi g\big\rangle_{\mu}\ge0$ for all $g\in L^{2}(\mathsf{E},\mu)$,
then $\Pi_{\nu}$ is positive, and consequently $\mathrm{Gap}(\Pi)=\mathrm{Gap}_{R}(\Pi)$
and $\mathrm{Gap}(\Pi_{\nu})=\mathrm{Gap}_{R}(\Pi_{\nu}).$ 
\end{enumerate}
Item \ref{item:gap-continuous-pi} is a restatement of \citep[Theorem 54]{andrieu-vihola-2012},
and \ref{item:augmented-positivity} follows because positivity of
$\Pi$ is equivalent to $\mathrm{Gap}_{L}(\Pi)\ge1$, implying $\mathrm{Gap}_{L}(\Pi_{\nu})\ge1.$
\end{rem}

\end{document}